\documentclass[aps,pra,reprint,floatfix,amsmath,amssymb,showkeys,nofootinbib,longbibliography]{revtex4-2}

\usepackage[english]{babel}
\usepackage[utf8]{inputenc}
\usepackage{amsmath}
\usepackage{amssymb}
\usepackage{amsthm}
\usepackage{mathrsfs}
\usepackage{graphicx}
\usepackage{booktabs}
\usepackage{xcolor}
\usepackage{tikz}
\usetikzlibrary{arrows.meta,shapes.geometric,positioning}
\usepackage[colorlinks,citecolor=blue,linkcolor=blue,urlcolor=blue,
            bookmarks=false]{hyperref}

\graphicspath{{./}}

\theoremstyle{plain}
\newtheorem{theorem}{Theorem}
\newtheorem{proposition}{Proposition}

\newtheorem{corollary}{Corollary}
\newtheorem{conjecture}{Conjecture}
\theoremstyle{definition}
\newtheorem{definition}{Definition}
\theoremstyle{remark}
\newtheorem{remark}{Remark}

\newcommand{\ket}[1]{\lvert #1 \rangle}

\newcommand{\ketbra}[2]{\lvert #1 \rangle\!\langle #2 \rvert}
\newcommand{\proj}[1]{\lvert #1 \rangle\!\langle #1 \rvert}
\newcommand{\tr}{\operatorname{tr}}
\newcommand{\id}{\mathbb{I}}

\newcommand{\PP}{\mathrm{P}}
\newcommand{\NP}{\mathrm{NP}}
\newcommand{\poly}{\mathrm{poly}}

\begin{document}

\title{Objectivity, basis ambiguity, and the classical simulability of\\
       concordant quantum computation}

\author{Bishal Kumar Das}
\email[]{bishal.9197@physics.iitm.ac.in}
\affiliation{Department of Physics, Indian Institute of Technology Madras,
             Chennai 600036, India}
\affiliation{Centre for Quantum Information, Communication and Computing,
             Indian Institute of Technology Madras, Chennai 600036, India}

\author{Sourav Manna}
\email[]{mannaphy@physics.iitm.ac.in}
\affiliation{Department of Physics, Indian Institute of Technology Madras,
             Chennai 600036, India}
\affiliation{Centre for Quantum Information, Communication and Computing,
             Indian Institute of Technology Madras, Chennai 600036, India}

\author{Vaibhav Madhok}
\email[]{madhok@physics.iitm.ac.in}
\affiliation{Department of Physics, Indian Institute of Technology Madras,
             Chennai 600036, India}
\affiliation{Centre for Quantum Information, Communication and Computing,
             Indian Institute of Technology Madras, Chennai 600036, India}

\date{\today}

\begin{abstract}
Concordant quantum computations, in which circuits whose register carries zero quantum
discord with respect to every subsystem at every step, contain neither
entanglement nor any other form of quantum correlation, yet no efficient
classical simulation of them is known beyond restricted gate sets. We argue
that the obstruction is not a shortage of correlations but the absence of a
uniquely identifiable local product basis, and we develop this \emph{basis
ambiguity} on two fronts. Quantitatively, we introduce the Basis Ambiguity
Measure $\beta\in[0,1]$, the normalized dimension of the manifold of admissible
diagonalizing bases on a gate support, and prove a sharp equivalence: a
concordance-preserving gate that defeats the Local Basis Finder of Cable and
Browne exists if and only if $\beta>0$. Both directions are constructive.
Physically, we identify the resource whose absence the obstruction is:
\emph{basis objectivity}, a redundant and independently accessible record of
the local basis. We show that an environment in spectrum broadcast form,
together with the system--environment interaction map, replaces the NP-hard
degeneracy test of Eastin by a fixed number of pairwise fragment comparisons
per gate, at a cost independent of register size and circuit depth; the two
ingredients live on different axes (state and dynamics), and neither alone
suffices. We further prove that a theory in which every state has a strictly
unique decomposition into perfectly distinguishable pure states has a simplex
state space and lacks the continuity content of Hardy's fifth axiom, so that
$\beta$ is a computational witness of that axiom.
\end{abstract}


\maketitle

\section{Introduction}
\label{sec:intro}

One of the central goals of quantum information science is to identify the
physical principles that distinguish genuinely quantum information processing
from its classical counterpart. Considerable progress has been made through the
identification of resources such as entanglement, quantum discord,
contextuality, and magic states. Despite these advances, the precise location of
the quantum--classical boundary remains an active area of research, particularly
in situations where conventional quantum correlations are absent.

For pure-state circuits the picture is comparatively settled. Jozsa and Linden
\cite{Jozsa_Linden} showed that a computation whose entanglement remains
confined to bounded blocks is efficiently simulable, and Vidal
\cite{PhysRevLett.91.147902} sharpened this into a graded statement: with
Schmidt rank $\chi$ across every cut the classical cost is
$\poly(n)\cdot\poly(\chi)$, so that entanglement $E_\chi = \log\chi$ growing
only logarithmically still permits efficient simulation. Entanglement is
therefore necessary for exponential speedup with pure states; the
Gottesman--Knill theorem \cite{gottesman1998heisenberg} shows that it is not
sufficient, and quantum search has been shown to be possible without
entanglement at exponential cost in other resources
\cite{PhysRevA.61.010301,Or_s_2004}. More generally, low entanglement is only
one of several structural properties that render a circuit classically
tractable, and a broad survey of the probabilistic methods available is given by
Van den Nest \cite{nest2010simulatingquantumcomputersprobabilistic}.

Mixed-state computation behaves differently
\cite{PhysRevLett.81.5672,ambainis2000computinghighlymixedstates}. In the DQC1
model the register remains separable, or nearly so, throughout, yet no efficient
classical simulation is known, and quantum discord
\cite{PhysRevLett.88.017901,ZurekPra2003,RevModPhys.84.1655,
PhysRevA.83.032323,mdtqc,doi:10.1142/S0217979213450410} rather than
entanglement has been proposed as the responsible resource
\cite{PhysRevLett.100.050502,Datta_2007}.

A striking example of how far this can be pushed is provided by
\emph{concordant} computation. Concordant states possess zero quantum discord
with respect to every subsystem and are therefore diagonal in a tensor product
of local bases. Since these states contain neither entanglement nor discord, one
might naturally expect them to admit efficient classical simulation. Eastin
\cite{eastin2010simulating} showed that a broad class of concordant
computations, namely those in a conventional form, with a product input state diagonal
in the computational basis, gates of support at most two qubits, and terminal
single-subsystem measurements, can indeed be simulated efficiently. But the
simulation requires more than the probability distribution associated with the
state: the simulator must also identify and track the local basis in which the
state remains diagonal throughout the computation.

This requirement reveals an unexpected feature of concordant dynamics. Although
the state carries only classical correlations, efficient simulation is
obstructed whenever the relevant local basis cannot be uniquely identified. The
obstacle arises from degeneracies in the density operator: when the branches of
the state are not distinguished by their weights (or, when the gate acts only
on part of the register, by the operators conditioned on the remainder), many
local product bases diagonalize the same state, and the simulation algorithm
faces a basis-selection problem that has no counterpart in classical probability
theory. The resulting obstruction is not associated with entanglement or
discord but with the non-uniqueness of basis descriptions permitted by the
geometry of quantum state space.

Remarkably, a closely related problem appears in the foundations of quantum
mechanics. The emergence of classical reality requires not only the suppression
of quantum coherence but also the emergence of an objective description of
physical systems. Decoherence theory explains how interactions with the
environment suppress interference between preferred pointer states
\cite{Zurek1981_pointer,Zurek1982,RevModPhys.75.715,zurek2003QtoC}, but
decoherence alone does not
explain why multiple independent observers can consistently agree on the state
of a system. Quantum Darwinism \cite{Ollivier2004-dv,Blume-Kohout2006-my,
Zurek2009,Zurek2014} extends the framework by showing how information about the
pointer states becomes redundantly encoded in the environment, so that many
observers can independently infer the same properties of a system without
disturbing it. The result is an operational notion of objectivity.

The observation from which this paper starts is that the basis-selection problem
of concordant simulation bears a direct resemblance to the preferred-basis
problem addressed by quantum Darwinism. In the former the question is: which
basis should the simulator track? In the latter: which basis becomes objectively
real? Both problems originate in the fact that a density matrix alone need not
determine a unique basis, and in both the essential challenge is selecting a
physically meaningful basis from a larger set of mathematically admissible
alternatives.

\subsection{The simulation obstruction}
\label{sec:whereitsits}

We first state the obstruction precisely, since it motivates the measure introduced below.

Eastin's algorithm tracks a probability distribution together with a local
basis, updating the basis at each gate. Its one expensive step is diagnosing
which computational basis labels lie in a common degenerate sector of the
current state, information the simulator does not hold directly and must
reconstruct from the accumulated circuit history. For gates of support one or
two the accumulated permutation is affine and the diagnosis reduces to $O(n)$
evaluations. For support three or more it does not: the permutation can
implement an arbitrary Boolean formula, and the test becomes equivalent to
deciding satisfiability. Eastin is careful to note that this rules out a direct
extension of his method rather than all methods.

Cable and Browne \cite{Cable_2015} circumvent that bottleneck by changing what
the algorithm reads. Their Local Basis Finder (LBF) works from the \emph{gate}
rather than from the state, so the NP-hard state-symmetry test never arises, and
they extend efficient simulation to a strictly larger class of computations,
including larger gates and qudits, with an emphasis on exactness and on the
arithmetic complexity of the equations involved. But the cost reappears
elsewhere. The LBF must sweep every computational basis projector on the gate
support, of every rank, and there are $2^{d_b}-1$ of these, where $d_b = d^{\,k}$
is the dimension of a support of $k$ subsystems; the overhead is therefore
doubly exponential in the number of subsystems the gate touches, and the
efficiency statements are correspondingly made at fixed maximum gate size. They
also exhibit gates for which the algorithm returns mutually incompatible answers
and halts.

The contrast with the pure-state case is instructive. A gate on
$k$ qubits that creates no entanglement is necessarily equivalent to a product
of single-subsystem unitaries and a permutation of the factors, so it costs only
$O(k)$ to apply, whatever $k$ is; the simulation cost is governed by the
entanglement rather than by the support. In the concordant sector, support is
precisely the parameter that matters, and the LBF pays $\approx 2^{d^{k}}$,
exponentially worse than writing the gate down.

The motivation of Eastin's work was to establish discord as the mixed-state
analogue of entanglement---to show that a computation generating no discord can
always be simulated classically, so that discord is necessary for advantage.
However, the proof was limited to showing that  a
concordant computation \emph{of conventional form and built from gates of
support at most two}  must generate discord in order to solve a classically
hard problem. For larger gate support the claim is not established.

 This motivates our
 work. If concordant computation with gates of larger support turns out
not to be efficiently simulable, then quantum advantage is available with no
entanglement, no discord, and no quantum correlation of any kind, and discord
cannot be the resource responsible. Something else must be. Eastin's own
methodological remark supplies the criterion: if computations lacking a property
can be simulated efficiently, then that property is a necessary resource for
advantage. What has been missing is a candidate for this property, and a graded
measure of it playing the role that $E_\chi$ plays for pure states.

\subsection{Contributions and organization}
\label{sec:contributions}

We make three contributions.

\emph{(i) A named resource.} We argue that efficient simulation requires not
merely the existence of a diagonal representation but the existence of a
uniquely identifiable one, and we call the property that supplies it
\emph{basis objectivity} (Definition~\ref{def:basis_objectivity}). We show that
it decomposes into two ingredients living on two different axes: the record
classes supplied by an environment in spectrum broadcast form, which is a
property of the \emph{state}, and the interaction map supplied by $H_{SE}$,
which is a property of the \emph{dynamics}. Neither alone suffices. To our
knowledge this state-level layer of objectivity has not previously been isolated
in the Darwinism literature.

\emph{(ii) A quantitative resolution of the obstruction.} In an oracle model in
which the environmental record is physically given
(Sec.~\ref{sec:oraclemodel}), Proposition~\ref{prop:darwinres} replaces the
intractable circuit-history test by $O(d^{2k}\poly(D_F))$ pairwise comparisons
of conditional fragment states per gate---independent of the register size $n$,
the circuit depth $T$, and the accumulated permutation.

\emph{(iii) A graded measure with a sharp threshold.} We introduce the Basis
Ambiguity Measure $\beta\in[0,1]$ (Definition~\ref{def:BAM}), the normalized
total multiplicity of the degenerate sectors of a concordant state on a gate
support, and prove that $\beta = 0$ if and only if \emph{no}
concordance-preserving gate on that support can defeat the LBF
(Theorems~\ref{thm:beta_zero} and \ref{thm:beta_pos}). Both directions are
constructive. We also prove that a theory in which every state has a strictly
unique decomposition into perfectly distinguishable pure states has a simplex
state space and lacks the continuity content of Hardy's fifth axiom
(Theorem~\ref{thm:hardy_collapse}), and that the normalization of $\beta$ is
Hardy's count of the parameters by which quantum theory exceeds classical
probability (Proposition~\ref{prop:hardy_excess}), so that $\beta$ is a
computational witness of that axiom.

We are equally explicit about the logical status of the central claim. Following
Eastin's framing, $\beta > 0$ is \emph{necessary} for the known simulation
strategies to fail at a given step, and is not claimed to be sufficient for
hardness; that stronger statement is Conjecture~\ref{conj:main}(ii) below and
remains open. We also record, in Sec.~\ref{sec:nonmono}, that the
\emph{frequency} with which the LBF actually fails behaves non-monotonically in
$\beta$; since that behavior is a property of one particular algorithm under one
particular distribution over gates, we draw no conclusion from it and state the
scaling of simulation cost with $\beta$ as an open problem.

The paper is organized as follows. Section~\ref{sec:concordant} reviews
concordant computation, the structure of concordance-preserving gates, and the
precise mechanism by which the degeneracy test becomes intractable.
Section~\ref{sec:hardy} traces the origin of basis ambiguity to Hardy's
continuity axiom. Section~\ref{sec:darwinism} develops basis objectivity and the
Darwinian oracle, and states the main conjecture. Section~\ref{sec:bam}
introduces $\beta$ and proves the sharp threshold.
Section~\ref{sec:resource} places $\beta$ in a resource-theoretic framework and
examines the boundary between generalized probabilistic theories.
Section~\ref{sec:dob} introduces Darwinism of the Basis, relates $\beta$ to
Hardy's axioms, and proves the axiomatic collapse theorem.
Section~\ref{sec:discussion} concludes. Three appendices collect
the LBF algorithm, an explicit two-qubit failure instance, and the sampling
procedure behind Sec.~\ref{sec:nonmono}.

\section{Concordant computation and the basis-identification problem}
\label{sec:concordant}

\subsection{Concordant states}
\label{sec:concordant_states}

Quantum discord, introduced independently by Ollivier and Zurek
\cite{PhysRevLett.88.017901} and by Zurek \cite{ZurekPra2003} and reviewed in
Ref.~\cite{RevModPhys.84.1655}, quantifies nonclassical correlations beyond
entanglement. For a bipartite state $\rho_{AB}$ the discord
$D(B|A)$ measures the discrepancy between two classically equivalent expressions
for the mutual information,
\begin{equation}
D(B|A) = I(\rho_{AB}) - J(B|A),
\end{equation}
with $I(\rho_{AB}) = S(\rho_A)+S(\rho_B)-S(\rho_{AB})$ and
$J(B|A) = S(\rho_B) - \min_{\{\Pi^A_j\}} S(\rho_{B|\{\Pi^A_j\}})$ the classical
correlation obtained by optimizing over projective measurements on $A$. A state
has zero discord exactly when all correlations are classical, i.e.\ when it can
be written as
\begin{equation}
\rho_{AB} = \sum_i p_i \proj{a_i} \otimes \proj{b_i}
\label{eq:zero_discord}
\end{equation}
for some orthonormal bases $\{\ket{a_i}\}$ and $\{\ket{b_i}\}$.

A \emph{concordant} state is one with zero discord with respect to every
subsystem \cite{eastin2010simulating}. For $N$ subsystems this takes the form
\begin{equation}
\rho = \sum_{k_1,\ldots,k_N} p(k_1,\ldots,k_N)\,
        \Pi^{(1)}_{k_1} \otimes \cdots \otimes \Pi^{(N)}_{k_N},
\label{eq:concordant_state_proj}
\end{equation}
where $\{\Pi^{(j)}_{k_j}\}$ are orthogonal rank-one projectors for each
subsystem $j$. We use $\Pi$ uniformly for rank-one projectors throughout.
Equivalently,
\begin{equation}
\rho = \sum_{\vec{x}} p(\vec{x})\, U \proj{\vec{x}} U^\dagger ,
\label{eq:concordant_state_unitary}
\end{equation}
where $U = U^{(1)} \otimes \cdots \otimes U^{(N)}$ is a local unitary.
Throughout, $d$ denotes the local Hilbert-space dimension of each subsystem
($d=2$ for qubits); all statements hold for qudits unless otherwise noted. A
\emph{concordant computation} is one in which the state remains concordant at
every step.

Two remarks on the definition will matter later. First, concordance is a
stringent condition: it demands the vanishing of quantum correlations under a
wide family of measures simultaneously \cite{RevModPhys.84.1655}, so the states
considered here sit at the extreme classical end of the mixed-state landscape.
Second, the condition is easy to state but not easy to verify. Deciding whether
a given density matrix is of the form \eqref{eq:concordant_state_proj},and
if so exhibiting the local bases, is the subject of a dedicated literature
\cite{PhysRevA.83.020101}, and computing discord in general is NP-complete
\cite{Huang_2014}. The difficulty of \emph{finding} the product basis, rather
than of certifying that one exists, is the thread that runs through this paper.

It is worth noting at the outset why the obvious simulation strategy fails.
Since the state at each step is a convex combination of product pure states, one
might try to track a single sampled branch through the circuit. But the
individual branches need not remain product states even when the mixture does.
Let $\rho$ be the equal mixture of
$\ket{\psi_{1,2}} = (\ket{0}\pm\ket{1})\ket{1}/\sqrt{2}$ and let the gate be
CNOT. The branches evolve into the two Bell-type vectors
$(\ket{01}\pm\ket{10})/\sqrt2$, which are maximally entangled, while the mixture
$\rho$ remains separable and concordant. Tracking pure branches is therefore not
an option; what must be tracked is the product basis itself. This is not an
accident of the example: Sec.~\ref{sec:entangled_mixture} shows that a
concordant state can equally well be built out of entangled pure states, and
identifies the feature of quantum theory responsible.

\subsection{Concordance-preserving gates}
\label{sec:gates}

The basis-identification problem originates in the structure of
concordance-preserving gates. The following proposition characterizes the most
general form of a gate mapping one concordant state to another. Its content is
due to Eastin \cite{eastin2010simulating}; we reproduce the argument in the
notation used below for completeness.

\begin{proposition}[Structure of concordance-preserving gates]
\label{prop:concordance_preserving}
Let
$\rho^{t-1} = \sum_{\vec{i}} p_{\vec{i}}\, \Pi^{t-1}_{\vec{i}}$
be a concordant state, where
$\Pi^{t-1}_{\vec{i}} = U_{t-1} \proj{\vec{i}} U_{t-1}^\dagger$
are rank-one product projectors defining the local basis. Suppose a gate $G_t$
acts such that $\rho^{t} = G_t \rho^{t-1} G_t^\dagger$ remains concordant. Then
\begin{equation}
G_t = U_t P_t B_t U_{t-1}^\dagger,
\label{eq:gate_degenerate}
\end{equation}
where $U_t$ is a local unitary, $P_t$ is a permutation matrix, and $B_t$ is
block-diagonal with respect to the eigenspaces of $\rho^{t-1}$. In particular:
\begin{enumerate}
\item if the spectrum of $\rho^{t-1}$ is non-degenerate, every block is
  one-dimensional and $B_t$ reduces to a diagonal phase matrix $D_t$ in the
  basis $\{\Pi^{t-1}_{\vec{i}}\}$, so that
  \begin{equation}
  G_t = U_t P_t D_t U_{t-1}^\dagger ;
  \label{eq:gate_nondegenerate}
  \end{equation}
  the factor $D_t$ acts trivially on $\rho^{t-1}$ but is not in general the
  identity, and cannot in general be absorbed into $U_t$, since a diagonal
  unitary in a product basis need not factorize as a tensor product of local
  unitaries;
\item if $\rho^{t-1}$ possesses degenerate eigenspaces, the corresponding blocks
  of $B_t$ have dimension greater than one and $B_t$ may act as an arbitrary
  unitary within them.
\end{enumerate}
\end{proposition}

\begin{proof}
Since $\rho^{t}$ is concordant there exists a local basis $U_t$ with
$\rho^{t} = \sum_{\vec{i}} p_{\vec{i}} \Pi^{t}_{\vec{i}}$ and
$\Pi^{t}_{\vec{i}} = U_t \proj{\vec{i}} U_t^\dagger$. Unitary evolution
preserves the spectrum, so $G_t$ can only rearrange eigenspaces carrying the
eigenvalues $\{p_{\vec{i}}\}$. If the spectrum is non-degenerate every
eigenspace is one-dimensional, and each projector must be mapped onto a unique
projector carrying the same eigenvalue,
\begin{equation}
G_t \Pi^{t-1}_{\vec{i}} G_t^\dagger = \Pi^{t}_{P_t(\vec{i})} .
\label{eq:projector_map}
\end{equation}
Equation~\eqref{eq:projector_map} determines the action of $G_t$ on each
one-dimensional eigenspace only up to a phase, since $\Pi^{t}_{P_t(\vec{i})}$ is
invariant under multiplication of the corresponding eigenvector by a phase.
Collecting these phases into a diagonal matrix $D_t$ gives
Eq.~\eqref{eq:gate_nondegenerate}. That the phase freedom cannot be dropped is
seen from the controlled-$Z$ gate, which is diagonal in the computational basis
and therefore concordance-preserving, yet is not of the form
$U_t P_t U_{t-1}^\dagger$ for any local $U_t,U_{t-1}$ and any permutation $P_t$.

If the spectrum contains degeneracies, several projectors correspond to the same
eigenvalue and the associated eigenspace is invariant under arbitrary unitary
rotations. Letting $B_t$ denote a unitary acting within each eigenspace, which
leaves the density matrix unchanged, the most general concordance-preserving
transformation is a permutation of eigenspaces together with an arbitrary
block-diagonal unitary acting inside them, which is
Eq.~\eqref{eq:gate_degenerate}.
\end{proof}

The distinction between the two cases is fundamental. In the non-degenerate case
the local basis is uniquely determined by the state, the residual freedom in
$D_t$ being physically inert. In the presence of degeneracy the additional
freedom represented by $B_t$ allows multiple local product bases to describe the
same density operator: the state no longer determines a unique basis.

\subsection{Degeneracy relative to the gate support}
\label{sec:degeneracy_support}

A clarification of the notion of degeneracy is needed, because it is used at two
different levels below. When the gate acts on the entire register, the branches
of the state are labeled by their probabilities alone, and two branches are
degenerate in the relevant sense exactly when those probabilities coincide:
ordinary spectral degeneracy. When the gate has support $b$ on only part of the
register, with remainder $a$, the state may be written as
\begin{equation}
\rho = \sum_{\vec{x}} \tilde{\rho}^{(a)}_{\vec{x}} \otimes \Pi^{(b)}_{\vec{x}},
\qquad
\tilde{\rho}^{(a)}_{\vec{x}} = p_{\vec{x}}\, \rho^{(a)}_{|\vec{x}},
\label{eq:conditional_form}
\end{equation}
and the role played by the eigenvalue of a branch is taken over by the
unnormalized conditional operator $\tilde{\rho}^{(a)}_{\vec{x}}$ on the
remainder. Two branches are degenerate when these operators coincide; we refer
to this as \emph{generalized}, or subspace, degeneracy, and it reduces to
ordinary spectral degeneracy when $b$ is the whole register. Coincidence of the
conditional operators implies equality of the corresponding probabilities, since
$p_{\vec{x}} = \tr \tilde{\rho}^{(a)}_{\vec{x}}$, but the converse does not hold:
branches of equal weight whose conditional operators differ confer no freedom on
the basis of $b$. Throughout we speak of degeneracy in this generalized sense,
and Sec.~\ref{sec:bam} makes it quantitative.

Although the probabilities associated with degenerate basis labels are
identical, the simulator must still determine \emph{which} labels belong to the
same degenerate sector before constructing the decomposition of the gate. The
identification of these sectors is the central basis-tracking problem.

\subsection{The simulation requirement and the degeneracy test}
\label{sec:simulation_requirement}

Eastin's strategy \cite{eastin2010simulating} rests on the following
observation. Since eigenvalues are preserved under unitary conjugation, the
effect of a gate $G_t$ on a concordant state is equivalent to a change of the
local product eigenbasis from $U_{t-1}$ to $U_t$ together with a permutation
$P_t$ of the eigenvalues, as in Eq.~\eqref{eq:gate_nondegenerate}. The
simulation then proceeds by Monte Carlo sampling: draw an initial basis label
$\ket{\vec{i}}$ with probability $p^0_{\vec{i}}$, apply the accumulated
permutation $P = P_T\cdots P_1$ to obtain $\vec{j}$, and measure in the final
basis $U_T$. This is efficient provided the decomposition can be found at each
step.

In the degenerate case the block-diagonal factor $B_t$ does not alter the
sampled distribution, since degenerate labels carry equal weights and \emph{any}
choice of basis within a degenerate sector yields the correct output
statistics, and the algorithm may absorb $B_t$ into the choice of $U_t$. What
the algorithm cannot avoid is determining which labels span a common degenerate
sector, since without this no valid decomposition of $G_t$ can be constructed at
all. The bottleneck is therefore not the probability evolution but the basis
tracking: finding $U_t$ such that $V = U_t^\dagger G_t U_{t-1}$ is a permutation
matrix.

The test is equivalent to asking whether the permutation $Q$ that swaps two
candidate labels commutes with $\rho^0$, given the accumulated circuit history.
With $P = P_{t-1}\cdots P_1$ and $A := P^\dagger Q P$, which satisfies
$A^2 = \id$ and hence $A^\dagger = A$, the condition reads
\begin{equation}
\rho^0 = A\, \rho^0 A .
\label{eq:degeneracy_check}
\end{equation}
The complexity of evaluating Eq.~\eqref{eq:degeneracy_check} depends critically
on the structure of $A$.

\begin{theorem}[Eastin \cite{eastin2010simulating}]
\label{thm:eastin}
For concordant circuits composed of gates acting on one or two qubits, the
accumulated permutation $P^\dagger Q P$ is an affine involution on binary
vectors, and the test \eqref{eq:degeneracy_check} reduces to $O(n)$ evaluations
on standard basis vectors. It is therefore efficient.
\end{theorem}

\begin{remark}[Hardness for larger gate support]
\label{rem:hardness}
For gates acting on three or more qubits, $P^\dagger Q P$ need not be affine:
the accumulated permutation can implement arbitrary Boolean functions on binary
vectors. Eastin shows that a general means of testing
Eq.~\eqref{eq:degeneracy_check} would imply the ability to solve 3-SAT, by
choosing $P$ to implement a Boolean formula, $Q$ to copy the result onto an
ancillary qubit, and $\rho^0$ to consist of unbiased input qubits together with
maximally biased ancillas, so that the test fails precisely when the formula is
satisfiable \cite{eastin2010simulating}. The test is therefore NP-hard. We note
the attribution carefully: the breakdown of the affine structure for support
three or more, and the resulting intractability, are observed by Eastin; the
reduction as formulated here in terms of Eq.~\eqref{eq:degeneracy_check} is our
restatement of that observation. Consequently no efficient extension of Eastin's
simulation to gates of \emph{arbitrary} support is known, although restricted
extensions are \cite{Cable_2015}.
\end{remark}

The key point is that the hardness is \emph{not} in choosing among valid $U_t$
once the degenerate sector is known, any choice within the sector gives a valid
permutation and the correct output distribution because degenerate labels carry
equal weights. The hardness lies entirely in \emph{identifying} the sector from
the accumulated circuit history encoded in $P$.

\subsection{The Local Basis Finder}
\label{sec:lbf}

Cable and Browne \cite{Cable_2015} approach the problem from the complementary
direction, focusing on the gate rather than on the state. Their Local Basis
Finder recovers $U_t$ from the gate alone. For each computational basis
projector $X_k$ on the support, of any rank, it forms
$\tilde{X}_k = G_t U_{t-1} X_k U_{t-1}^\dagger G_t^\dagger$ and, for each
subsystem $j$ of the support, solves
\begin{equation}
\bigl[\, \id^{(b\setminus j)} \otimes \Pi^{(j)},\ \tilde{X}_k \,\bigr] = 0
\label{eq:commutator}
\end{equation}
for rank-one $\Pi^{(j)}$. Three outcomes are possible for a given $X_k$: no
complete local basis exists, because $\tilde{X}_k$ is entangled, and $X_k$ is
discarded; a unique local basis is determined; or a local basis exists but is
not unique, in which case only the constraints common to all solutions are
retained, so that an uninformative $X_k$ contributes nothing. The algorithm
halts with failure only when two retained constraints are mutually
\emph{incompatible}---when a projector demanded by one $X_k$ fails to commute
with $\tilde{X}_{k'}$ for another. Non-uniqueness for an individual $X_k$ is
anticipated and harmless; incompatibility across different $X_k$ is the failure
condition. We refer to a gate producing incompatibility as \emph{defeating} the
LBF. The algorithm is summarized in Appendix~\ref{app:lbf} and an explicit
two-qubit gate defeating it is given in Appendix~\ref{app:lbf_failure}.

\subsection{What the two approaches have in common}
\label{sec:common}

The LBF does not explicitly test for degeneracy in the pre-gate state; it
attempts to reconstruct the decomposition of $G_t$ directly by tracking how
computational basis projectors evolve. This reformulation is powerful, and it
buys a strictly larger simulable class. But in the presence of the
block-diagonal unitary $B_t$, which by Eq.~\eqref{eq:gate_degenerate} is
physically allowed precisely when the pre-gate state possesses degeneracy, the
direct decomposition of $G_t$ also becomes inefficient and generally fails to
yield a unique solution.

Both approaches are therefore searching for the same artifact: the local product
basis $U_t$ in which the state remains concordant. The search becomes hard for
the same underlying reason. For Eastin it manifests as the algorithmic hardness
of identifying the degenerate sectors of the pre-gate state; for Cable and
Browne it manifests as the unresolvable ambiguity introduced by $B_t$ in the
gate decomposition. In both paradigms it is the continuous basis freedom inside
degenerate sectors that derails efficient classical simulability.

The resulting picture is as follows - Degeneracy on a gate support is what
permits a concordance-preserving gate to carry a nontrivial block-diagonal
factor; the local product basis is then no longer fixed by the state alone, and
the admissible bases form a continuous manifold rather than a single point. When
the degeneracy is absent, or has been resolved by information from outside the
state, that manifold collapses to a point, a unique basis is singled out, and
the step is efficiently simulable by basis tracking. Degeneracy is in this sense
the source of the residual quantum character of concordant computation, even
though such states carry neither entanglement nor discord.

This suggests that the obstruction is not an artifact of either algorithm but a
property of quantum state space itself, and that the missing ingredient can be
named. Sec.~\ref{sec:hardy} traces the
freedom to Hardy's continuity axiom, and Sec.~\ref{sec:darwinism} identifies
what must be supplied to remove it.

\section{The origin of basis ambiguity: Hardy's continuity axiom}
\label{sec:hardy}

The basis ambiguity described above is not an accident. It is a direct
consequence of the structure of quantum theory as distinguished from classical
probability theory by Hardy's axioms \cite{hardy2001quantum}.

Hardy showed that the only difference between quantum and classical physics,
viewed as theories of probability over pure states, is that quantum theory
permits continuous reversible transformations between pure states while
classical theory does not. In classical probability theory pure states form a
discrete, perfectly distinguishable (set the vertices of a simplex), and any
mixed state admits a unique decomposition into pure states. In quantum theory
the state space is not a simplex, and degenerate mixed states admit continuously
many equivalent decompositions.

\subsection{Classical decompositions are unique}

In classical probability theory any two distinct pure states are perfectly
distinguishable in a single shot, since a pure state is a complete specification
of the system. In the framework of generalized probabilistic theories (GPTs)
this is the statement that a classical state space is a simplex whose vertices
are the pure states \cite{PhysRevA.75.032304}; in Hilbert-space language,
perfect single-shot distinguishability is orthogonality.

\begin{corollary}
\label{cor:classical_basis}
In classical probability theory the pure-state decomposition of any
state---and hence the analogue of the product basis underlying a concordant
state---is unique.
\end{corollary}

\begin{proof}
In the GPT framework \cite{PhysRevA.75.032304} a classical system is represented
by a simplex whose vertices are its pure states. A point in a simplex has a
unique convex
decomposition into vertices; hence any classical mixed state determines its
pure-state decomposition uniquely, and there is no freedom in the choice of
basis.
\end{proof}

Quantum state spaces are not simplices: non-orthogonal pure states exist within
degenerate subspaces, and a degenerate mixed state admits infinitely many
equivalent decompositions related by continuous unitaries acting within that
subspace.

\subsection{Concordant states can be prepared from entangled ones}
\label{sec:entangled_mixture}

The cleanest way to see what the continuity axiom buys, and why it bears on
simulation, is to exhibit the non-uniqueness explicitly rather than to argue
from the geometry. A concordant state carries no entanglement across any cut.
It can nevertheless \emph{also} be obtained as a mixture of pure states that are
themselves entangled, the entanglement of the parts leaving no trace in the
whole. The following proposition records when this is possible and what form it
takes; the two regimes it distinguishes are discussed below.

\begin{proposition}[Entangled preparations of a concordant state]
\label{prop:entangled_mixture}
Let $\rho = \sum_{\vec{x}} p_{\vec{x}} \Pi_{\vec{x}}$ be concordant on $N \ge 2$
subsystems, with $\Pi_{\vec{x}} = U\proj{\vec{x}}U^\dagger$ and
$U = \bigotimes_j U^{(j)}$ local, and let
$S = \{\vec{x} : p_{\vec{x}} > 0\}$. Fix $\vec{x} \neq \vec{y}$ in $S$ and
$\theta \in \mathbb{R}$, and define the unnormalized vectors
\begin{align}
\ket{\varphi_1} &= \phantom{-}\cos\theta\,\sqrt{p_{\vec{x}}}\,U\ket{\vec{x}}
                 + \sin\theta\,\sqrt{p_{\vec{y}}}\,U\ket{\vec{y}},
\nonumber\\
\ket{\varphi_2} &= -\sin\theta\,\sqrt{p_{\vec{x}}}\,U\ket{\vec{x}}
                 + \cos\theta\,\sqrt{p_{\vec{y}}}\,U\ket{\vec{y}}.
\label{eq:phi_pair}
\end{align}
Then, for every $\theta$,
\begin{equation}
\rho = \sum_{\vec{z} \in S\setminus\{\vec{x},\vec{y}\}} p_{\vec{z}}\Pi_{\vec{z}}
     \;+\; \ketbra{\varphi_1}{\varphi_1} + \ketbra{\varphi_2}{\varphi_2},
\label{eq:alt_decomp}
\end{equation}
and
\begin{equation}
\langle \varphi_1 | \varphi_2 \rangle
   = \tfrac{1}{2}\bigl(p_{\vec{y}} - p_{\vec{x}}\bigr)\sin 2\theta .
\label{eq:phi_overlap}
\end{equation}
Moreover, for $\theta \notin \tfrac{\pi}{2}\mathbb{Z}$ the vectors
$\ket{\varphi_{1,2}}$ have Schmidt rank two across exactly those bipartitions
placing at least one index in which $\vec{x}$ and $\vec{y}$ differ on each side.
In particular they are entangled if and only if $\vec{x}$ and $\vec{y}$ differ
in at least two subsystem indices.
\end{proposition}

\begin{proof}
Equation~\eqref{eq:alt_decomp} is a direct computation: the diagonal terms
contribute $(\cos^2\theta + \sin^2\theta)p_{\vec{x}} = p_{\vec{x}}$ and
likewise for $p_{\vec{y}}$, while the cross terms carry opposite signs and
cancel. Equation~\eqref{eq:phi_overlap} is immediate from orthonormality of
$\{U\ket{\vec{x}}\}$.

For the entanglement claim, local unitaries do not change entanglement, so take
$U = \id$. Let $A|B$ be a bipartition and write
$\ket{\vec{x}} = \ket{\vec{x}_A}\ket{\vec{x}_B}$. If every index in which
$\vec{x}$ and $\vec{y}$ differ lies in $A$, then $\vec{x}_B = \vec{y}_B$ and
\begin{equation}
\ket{\varphi_1} = \bigl(\cos\theta\sqrt{p_{\vec{x}}}\ket{\vec{x}_A}
   + \sin\theta\sqrt{p_{\vec{y}}}\ket{\vec{y}_A}\bigr)\otimes\ket{\vec{x}_B},
\end{equation}
a product vector across that cut, and similarly for $\ket{\varphi_2}$. If
instead differing indices occur on both sides then
$\ket{\vec{x}_A} \perp \ket{\vec{y}_A}$ and
$\ket{\vec{x}_B} \perp \ket{\vec{y}_B}$, so the two terms of
Eq.~\eqref{eq:phi_pair} already constitute a Schmidt decomposition with two
nonzero coefficients. A vector is a product vector overall precisely when it is
product across every bipartition, and a cut separating two differing indices
exists precisely when $\vec{x}$ and $\vec{y}$ differ in at least two of them.
\end{proof}

Two regimes must be distinguished. When the
two weights differ, Eq.~\eqref{eq:phi_overlap} is nonzero: the entangled
preparation exists, but as a non-orthogonal ensemble. It is an alternative way
of \emph{building} $\rho$, not an alternative basis in which $\rho$ is diagonal,
and a simulator that has fixed on the eigenbasis is not troubled by it.

When the weights coincide, that is, under degeneracy, the overlap vanishes and
the preparation is orthogonal. The state is then diagonal in a basis of
entangled vectors and equally diagonal in a basis of product vectors, with
nothing in $\rho$ to choose between them. The maximally mixed state of two
qubits is the familiar instance: it is at once an equal mixture of the four
computational product states and an equal mixture of the four Bell states.
Degeneracy is what promotes the freedom of
Proposition~\ref{prop:entangled_mixture} from a re-preparation of the same state
into a genuine ambiguity of basis, and this case is the subject of the remainder of the paper.

What the ambiguity looks like depends on how far apart the tied labels are, and
two qubits already show both cases. If $\ket{00}$ and $\ket{11}$ carry equal
weight, every rotation inside $\mathrm{span}\{\ket{00},\ket{11}\}$ produces
entangled vectors, and the only product vectors in that span are $\ket{00}$ and
$\ket{11}$ themselves: the state acquires a continuum of entangled eigenbases,
but its product eigenbasis remains unique. If instead $\ket{00}$ and $\ket{01}$
are tied, the sector is $\ket{0}\otimes\mathbb{C}^2$ and every rotation inside it
is a rotation of the second qubit conditioned on the first. The eigenvectors stay
product, and the basis of the second qubit within that branch becomes a
continuum; this becomes an ambiguity of the second qubit's local frame only if
the tie holds in every branch of the first. A permutation converts one case into
the other. The CNOT with the second qubit as control maps the pair
$\{00, 01\}$ onto $\{00, 11\}$, so a classical permutation applied after a
rotation inside a distance-one sector acts as an entangler on each individual
run, while the ensemble stays concordant. Applied the other way, it turns an
entangled ambiguity, invisible to any search over product bases, into a product
one that such a search cannot resolve. The CNOT example of
Sec.~\ref{sec:concordant_states} is the first of these mechanisms, and
Theorem~\ref{thm:beta_pos} uses the second to construct a gate that defeats basis
tracking.

Such a preparation exists for every concordant state in which at least two
subsystems are not in a definite state; the only exceptions are states in which
every subsystem but one is fixed.

Underlying both regimes is Hardy's fifth axiom. In a simplex, by
Corollary~\ref{cor:classical_basis}, a state has exactly one decomposition into
pure states, so a state prepared as a mixture of perfectly distinguishable
product pure states admits no alternative at all, entangled or otherwise, with
equal weights or unequal. The family \eqref{eq:phi_pair} is continuous in
$\theta$ and exists only because the pure states of quantum theory are not a
discrete, perfectly distinguishable set; delete the continuity and it collapses
to $\theta \in \tfrac{\pi}{2}\mathbb{Z}$, at which $\ket{\varphi_{1,2}}$ are
again the original product vectors. That a state with no entanglement can be
built from entangled ones is therefore not a fact about entanglement but a
signature of the continuity axiom, and degeneracy is the condition under which
that signature becomes an ambiguity a simulator must resolve.

Two consequences for simulation follow.

The first is the failure of the naive strategy noted in
Sec.~\ref{sec:concordant_states}: which pure states the mixture consists of is
not determined by $\rho$, so a simulator that samples a branch is sampling from
an ensemble it has chosen rather than one the state supplies, and some
admissible choices are entangled. The register must be tracked as a \emph{basis}
rather than as a branch. The CNOT example of Sec.~\ref{sec:concordant_states} is
the instance $\vec{x} = 01$, $\vec{y} = 10$, $\theta = \pi/4$ of
Eq.~\eqref{eq:phi_pair}.

The second is that basis tracking is itself obstructed once the degenerate
regime is reached, and that this is a fact about the theory rather than about
any algorithm: the freedom is absent in a theory whose state space is a simplex,
regardless of what algorithm one used. This is made precise in
Sec.~\ref{sec:gpt} for discrete GPTs and in Theorem~\ref{thm:hardy_collapse} in
general, and quantified in Sec.~\ref{sec:bam} by the measure $\beta$, which
counts exactly the orthogonal freedom of the degenerate regime, generalized from
weights to conditional operators so that it applies when the gate acts on part
of the register. Since the freedom is real, removing it requires supplying
information from outside the state. That information is the subject of the next
section.

\section{Basis objectivity and the Darwinian basis-selection oracle}
\label{sec:darwinism}

\subsection{The preferred-basis problem}

The preferred-basis problem has a long history in quantum foundations. A density
matrix may admit many equivalent diagonal representations. Decoherence theory
addresses this by showing how the system--environment interaction selects
certain states, the pointer states, that are stable under environmental
monitoring \cite{Zurek1981_pointer,Zurek1982,RevModPhys.75.715,zurek2003QtoC}.
The interaction
Hamiltonian $H_{SE}$ determines which system observables commute with, or are
least perturbed by, the environmental interaction, and these observables define
the pointer basis.

Decoherence alone, however, selects the pointer states at the level of the
system. It does not explain why multiple independent observers, each accessing
different parts of the environment, all agree on the same basis.

\subsection{Quantum Darwinism, strong quantum Darwinism, and spectrum
            broadcast structure}
\label{sec:definitions}

Because several inequivalent formalizations of objectivity coexist in the
literature, we fix conventions precisely. Let $S$ be the system and let the
environment $E = E_1 \otimes \cdots \otimes E_M$ be partitioned into fragments;
write $F \subseteq E$ for a fragment comprising a fraction $f$ of the
environment.

Quantum Darwinism (QD) \cite{Zurek2009,Zurek2014} holds, with error $\delta$ and
redundancy $R_\delta$, when the quantum mutual information between the system
and a typical fragment satisfies
\begin{equation}
I(S{:}F) \geq (1-\delta)\, H(S)
\quad \text{for typical } F \text{ of size } f_\delta ,
\label{eq:qd_condition}
\end{equation}
where $H(S)$ is the system's von Neumann entropy and $R_\delta = 1/f_\delta$
counts the number of disjoint fragments from which nearly complete classical
information about the pointer observable can be extracted independently
\cite{Ollivier2004-dv,Zurek2009,Blume-Kohout2006-my,Brandao2015}. The
characteristic signature is the ``classical plateau'' of $I(S{:}F)$ at $H(S)$
over a wide range of fragment sizes.

That this behavior is generic rather than contrived matters for what follows,
since the oracle of Sec.~\ref{sec:oraclemodel} presupposes an environment that
actually carries redundant records. The plateau has been established in a range
of concrete physical settings: photons scattering off an illuminated object
\cite{JessPrl2010,JessRiedelNjp2011}, spin environments
\cite{ZwolakNature2016}, and environments that are themselves mixed rather than
initially pure \cite{ZwolakPrl2009}; the last being the case relevant here,
since a concordant register is a mixed state throughout. Redundancy is also
known to be robust against environment fragments that carry no information about
the pointer observable \cite{ZwolakPra2017}, though it is not permanent: it
rises during decoherence and can subsequently decay as the environment
internally equilibrates \cite{JessRiedel_2012}. Brand{\~a}o, Piani, and
Horodecki \cite{Brandao2015} show that the emergence of such records is generic
for a broad class of dynamics.

A strictly stronger, structural notion is the \emph{spectrum broadcast
structure} (SBS) of Korbicz, Horodecki, and Horodecki
\cite{korbicz2013objectivitystatebroadcastingorigins}, in which the joint
system--environment state, after tracing out any unobserved part of the
environment, takes the form
\begin{equation}
\rho_{S E_1 \cdots E_M} = \sum_i p_i \proj{i}_S \otimes
   \rho^{(i)}_{E_1} \otimes \cdots \otimes \rho^{(i)}_{E_M},
\label{eq:sbs_condition}
\end{equation}
with $\rho^{(i)}_{E_k} \perp \rho^{(j)}_{E_k}$ for $i \neq j$, where
$\{\ket{i}\}$ is the pointer basis and $\perp$ denotes mutually orthogonal
supports. Orthogonality of the conditional fragment states means that each
observer, measuring her own fragment, can determine the pointer label $i$
perfectly and independently, without disturbing the system.

Between the two lies \emph{strong} quantum Darwinism. The mutual information in
Eq.~\eqref{eq:qd_condition} decomposes into a classically accessible part and a
genuinely quantum part,
\begin{equation}
I(S{:}F) = \chi(S{:}F) + D(S|F),
\label{eq:mi_decomposition}
\end{equation}
where $\chi(S{:}F)$ is the Holevo quantity of the ensemble the fragment holds
about the pointer observable, i.e. the classical information actually extractable
from $F$ by measurement, and $D(S|F)$ is the quantum discord. The two terms are
complementary: what the fragment holds as accessible classical information about
the pointer observable is exactly what it does not hold as discord
\cite{ZurekNature2013}. The plateau condition constrains only the sum, and can
therefore be satisfied even when a large part of $I(S{:}F)$ is discord
\cite{PhysRevA.98.032103}, in which case observers cannot actually read the
pointer label from their fragments. Strong
QD \cite{Le2019-ww} closes this loophole by requiring the plateau to be entirely
classical,
\begin{equation}
\chi(S{:}F) = H(S), \qquad D(S|F) = 0 ,
\label{eq:strong_qd}
\end{equation}
and adding \emph{strong independence}, (conditional independence of the
environment fragments given the pointer label, so that they share no
correlations beyond their common information about $S$) recovers SBS, the two
being equivalent by a theorem of Le and Olaya-Castro \cite{Le2019-ww}. The three
conditions therefore stand in the strict hierarchy
\begin{equation}
\text{SBS} \subsetneq \text{strong QD} \subsetneq \text{QD},
\label{eq:hierarchy}
\end{equation}
the second inclusion being strict because states exhibiting the
mutual-information plateau need not have broadcast structure; see
Ref.~\cite{Korbicz2021roadstoobjectivity} for a review. We note for later use
that all three are properties of the joint state $\rho_{SE}$ alone.

\paragraph{Convention.}
All results below assume SBS, Eq.~\eqref{eq:sbs_condition}, and not merely the
redundancy condition~\eqref{eq:qd_condition}. When we speak of a ``Darwinian
interaction'' we mean a system--environment dynamics whose output satisfies SBS
with respect to the concordant basis at each computational step. This is the
natural setting for our purposes: it is the orthogonality structure of the
fragment states, not the mutual-information plateau alone, that a simulator can
exploit.

\paragraph{Degeneracy and the broadcast state.}
Equation~\eqref{eq:sbs_condition} deserves a comment in the presence of
degeneracy, which is the situation of interest here. Note first that it is a
statement about \emph{mixed} states and is therefore directly applicable to
concordant dynamics; the familiar pure branching form
$\ket{i}\ket{E_0} \to \ket{i}\ket{E^{(1)}_i}\cdots\ket{E^{(M)}_i}$ is the
special case from which Eq.~\eqref{eq:sbs_condition} arises upon decoherence of
the branches. Second, whether two pointer labels $i \neq j$ with
\emph{accidentally equal} probabilities $p_i = p_j$ are recorded distinguishably
in the environment is a property of the interaction Hamiltonian $H_{SE}$, not of
the spectrum of $\rho_S$. Standard SBS requires orthogonal fragment states for
all distinct pointer labels. If instead $H_{SE}$ fails to resolve a degenerate
sector, the corresponding conditional fragment states coincide,
$\rho^{(i)}_{E_k} = \rho^{(j)}_{E_k}$, and the environment carries a record only
of the sector, not of any basis within it. Both situations occur physically, and
as shown in Sec.~\ref{sec:oracle_resolution} both are equally useful to the
simulator. We therefore do \emph{not} modify the definition of SBS; we only read
off which of the two obtains, fragment pair by fragment pair.

\subsection{Labels versus states}
\label{sec:labels}

It is essential to be precise about \emph{what} Eq.~\eqref{eq:sbs_condition}
makes objective. An observer holding a fragment $F$ can, by measurements on her
own fragment alone, resolve the orthogonal family $\{\rho^{(i)}_F\}$ and thereby
determine the index $i$ in every run and, over an ensemble, the probabilities
$p_i$, that is, the \emph{spectrum}. The terminology ``spectrum broadcast'' is
exact on this point: what the environment redundantly records is precisely the
spectrum, namely (label, probability) pairs together with the coincidence
structure among labels. What the fragment does \emph{not} contain is the
correspondence $i \mapsto \ket{i} \in \mathcal{H}_S$. The fragment data are
invariant under any relabeling on the system side: the outcome $i=3$ is equally
consistent with the system occupying $\ket{3}$ or $U\ket{3}$ for any local
unitary $U$. That correspondence resides in the interaction Hamiltonian
$H_{SE}$, or equivalently in the branching channel, that is, in the
\emph{dynamics}, and is not itself redundantly recorded in the state.

We shall say that SBS delivers \emph{label objectivity}: all observers agree on
the index, on the probabilities, and on which labels are (un)resolved, and each
observer can predict every other observer's outcome. The strictly stronger
condition in which observers can identify the basis vectors themselves,
\emph{state objectivity}, requires in addition knowledge of the interaction map.
The distinction matters computationally: a simulator needs $U_t$ as an explicit
matrix, and labels alone do not provide it.

\paragraph{Terminology.}
Three notions are in play and it is convenient to collect them. \emph{Label
objectivity} is what a condition on the joint state $\rho_{SE}$ alone can
deliver: the index, the probabilities, and the coincidence structure among
labels. \emph{State objectivity} additionally supplies the correspondence
$i \mapsto \ket{i}$ and is therefore a condition on the pair (state, dynamics).
\emph{Basis objectivity} (Definition~\ref{def:basis_objectivity} below) is state
objectivity for the concordant product basis, and is the computational resource
of this work. Its quantitative counterpart is the Basis Ambiguity Measure
$\beta$ of Sec.~\ref{sec:bam}, which measures the freedom remaining when basis
objectivity is absent.

\subsection{Relation to prior notions of objectivity}
\label{sec:prior_notions}

The label/state distinction refines, but does not contradict, notions of
objectivity already present in the literature. Brand{\~a}o, Piani, and Horodecki
\cite{Brandao2015} distinguish \emph{objectivity of outcomes} (observers agree
on results) from \emph{objectivity of observables}: different observers probing
fragments of the environment can only learn information extractable from a
single observable of the system which---although possibly non-unique---is
independent of which fragment is probed. Their theorem, extended to infinite
dimensions by Knott \emph{et al.} \cite{PhysRevLett.121.160401} and strengthened
by Qi and Ranard \cite{Qi2021emergent}, shows that any system-to-fragments
channel is, on most fragments, close to a measure-and-prepare channel with one
fixed measurement. This is the closest existing notion to basis objectivity, but
on inspection it is an \emph{existence} statement at the label level: all
observers are dynamically locked onto a common observable and can verify the
mutual consistency of their records, yet the identification of that observable's
eigenvectors in $\mathcal{H}_S$ is defined relative to the channel and is not
recoverable from the fragment data alone. Similarly, the intersubjectivity
results of Horodecki, Korbicz, and Horodecki \cite{Horodecki2015-rt} and the SBS
equivalence theorems \cite{Le2019-ww,Korbicz2021roadstoobjectivity} operate at
the level of agreement on outcomes.

Two consequences follow. First, the hierarchy \eqref{eq:hierarchy} is
one-dimensional: all three conditions are properties of the joint state
$\rho_{SE}$, and a property of the state alone can deliver at most label
objectivity. State objectivity is therefore not a further rung in this chain but
a condition on the \emph{pair} (state, dynamics). The resulting picture is
two-dimensional (Fig.~\ref{fig:twoaxes}): agreement strength along one axis
(QD $\to$ strong QD $\to$ SBS) and epistemic content along the other (label
$\to$ state). The computational resource identified here sits at the corner of
the two: SBS supplies the record classes, the interaction map supplies $U_t$,
and neither alone suffices.

Second, the need for side information to identify a basis connects naturally to
the theory of quantum reference frames \cite{RevModPhys.79.555}: specifying a
basis of $\mathcal{H}_S$ is reference-frame information that must be physically
shared between parties and cannot be conveyed by the relabeling-invariant record
structure alone. To our knowledge the isolation of this state-level layer of
objectivity, its physical location, and what its computational
significance is, has not previously been made explicit.

\begin{figure*}[t]
\centering
\resizebox{0.92\textwidth}{!}{
\begin{tikzpicture}[>={Stealth[length=2.4mm]}, font=\small]
\fill[gray!13] (-0.4,-0.65) rectangle (13.3,1.10);
\fill[blue!8]  (-0.4,2.32) rectangle (13.3,3.77);
\node[anchor=west, gray!45!black] at (-0.15,0.88)
  {\scriptsize \textbf{label objectivity}: index $i$, spectrum $\{p_i\}$,
   coincidence structure --- reachable from $\rho_{SE}$ alone};
\node[anchor=west, blue!55!black] at (-0.15,3.49)
  {\scriptsize \textbf{state objectivity}: basis vectors
   $\lvert i\rangle \in \mathcal{H}_S$ --- requires the dynamics};
\draw[dashed, thick, gray!60!black] (-0.4,1.62) -- (13.3,1.62);
\node[anchor=west, gray!45!black] at (0.5,1.94)
  {\scriptsize no condition on $\rho_{SE}$ alone crosses this line};
\draw[->, thick] (-0.4,-1.25) -- (13.5,-1.25);
\draw[->, thick] (-0.4,-1.25) -- (-0.4,4.10);
\node[below] at (6.1,-1.35)
  {agreement strength \; (conditions on the joint state $\rho_{SE}$)};
\node[rotate=90] at (-0.95,1.35) {epistemic content};
\node[circle, fill=black, inner sep=1.7pt, label=below:{QD}]
  (qd) at (1.6,0.05) {};
\node[circle, fill=black, inner sep=1.7pt, label=below:{strong QD}]
  (sqd) at (5.4,0.05) {};
\node[circle, fill=black, inner sep=1.7pt, label=below:{SBS}]
  (sbs) at (9.2,0.05) {};
\draw[->] (qd) -- node[above=1pt]{\scriptsize $D(S|F)\to 0$} (sqd);
\draw[->] (sqd) -- node[above=1pt]{\scriptsize $+$ strong independence} (sbs);
\node[star, star points=5, star point ratio=2.2, fill=blue!65,
      inner sep=2.1pt] (corner) at (9.2,2.92) {};
\node[anchor=east, align=right, blue!55!black] at (8.8,2.92)
  {\scriptsize \textbf{basis objectivity} (Def.~\ref{def:basis_objectivity})\\
   \scriptsize $=$ the computational resource
   (Prop.~\ref{prop:darwinres})};
\draw[->, thick, blue!65]
  (sbs) -- node[right=4pt, pos=0.74, align=left]
  {\scriptsize $+$ interaction map $i \mapsto \lvert i\rangle$\\
   \scriptsize (a property of $H_{SE}$)} (corner);
\end{tikzpicture}
}
\caption{The two-axis structure of objectivity. Horizontal axis: agreement
strength, ordered by the hierarchy
$\mathrm{QD} \supsetneq \mathrm{strong\ QD} \supsetneq \mathrm{SBS}$; every
condition on this axis is a property of the joint state $\rho_{SE}$ and
therefore delivers at most label objectivity (lower band). Vertical axis:
epistemic content. Crossing from label to state objectivity (upper band)
requires the interaction map $i \mapsto \lvert i\rangle$, which resides in the
dynamics $H_{SE}$ and is not redundantly recorded in the state. The corner, SBS
combined with the interaction map, is basis objectivity
(Definition~\ref{def:basis_objectivity}), the resource that resolves the
concordant-simulation obstruction (Proposition~\ref{prop:darwinres}): the record
classes come from the horizontal axis, $U_t$ from the vertical one, and neither
alone suffices.}
\label{fig:twoaxes}
\end{figure*}

\subsection{The oracle access model}
\label{sec:oraclemodel}

Before stating what the Darwinian record provides to a simulator we must be
precise about the computational model, since it differs from the setting of
Eastin's original problem. In the \emph{standard model} the simulator receives
only a classical description of the circuit $(U_0, G_1,\ldots,G_T)$ and the
initial distribution, and must sample from the output distribution in time
polynomial in $n$ and $T$. The hardness of Remark~\ref{rem:hardness} lives in
this model.

We consider instead \emph{simulation with a Darwinian oracle}. At each step $t$,
in addition to the circuit description, the simulator may query an oracle that
returns the conditional fragment states $\{\rho^{(i)}_F\}$ of
Eq.~\eqref{eq:sbs_condition} restricted to the support of the upcoming gate
$G_t$, for a single fragment $F$ of bounded dimension $D_F$, \emph{together with
the identification of the record labels $i$ with the simulator's basis labels in
the $U_{t-1}$ frame}; equivalently, the interaction map determined by $H_{SE}$
is known. In the terminology of Sec.~\ref{sec:labels} the oracle grants state
objectivity, not merely label objectivity. The record certifies the basis at the
time of imprinting; if the system evolves after decoupling from the environment
the correspondence $i \mapsto \ket{i}$ goes stale, and in our setting this is
handled by querying the oracle afresh at every gate step.

Physically this models an observer with access to a small piece of the
environment of a Darwinian process that monitors the register in the concordant
basis, \emph{and} who knows the monitoring interaction. The assumption that a
single fragment of bounded dimension already resolves the record is exactly the
high-redundancy regime that concrete models exhibit, a few scattered photons
suffice in the illumination models of
Refs.~\cite{JessPrl2010,JessRiedelNjp2011}, so $D_F$ in
Eq.~\eqref{eq:darw_cost} is a small constant rather than a quantity growing
with $n$. We emphasize that this
changes the computational task: if the environment itself had to be simulated
classically from first principles, computing the fragment states could be as
hard as the original degeneracy test and nothing would be gained. The content of
the result below is therefore a statement about \emph{what information the
physical environment holds, and how little of it suffices}, not a resolution of
Eastin's open problem in the standard model.

\subsection{What the Darwinian record provides}
\label{sec:oracle_resolution}

The simulator requires a valid local basis at each gate step. The SBS record
provides it through two ingredients.

\emph{The record structure.} Within an SBS state, pairwise comparison of the
conditional fragment states sorts the basis labels on the gate support into
\emph{record classes}: labels whose fragment states are mutually orthogonal are
distinguished by the environment, while labels whose fragment states coincide
are not. Either verdict resolves the simulator's problem---orthogonal records
single out a basis, coincident records certify a genuine degeneracy within which
any basis choice is valid.

\emph{The interaction map.} The Hamiltonian $H_{SE}$ provides the correspondence
between record-class labels and physical system states. This is what converts
the abstract record structure into an identification of $U_t$.

We formalize this as a proposition rather than a theorem: its proof is
elementary once the access model of Sec.~\ref{sec:oraclemodel} is granted, and
the substantive content lies in the model itself.

\begin{proposition}[Darwinian resolution of the degeneracy test]
\label{prop:darwinres}
Let $k = |G_t|$ denote the support size of the gate $G_t$, let $d$ be the local
Hilbert-space dimension, and suppose the system--environment state at step $t$
satisfies SBS, Eq.~\eqref{eq:sbs_condition}, with respect to the concordant
basis, with fragment dimension $D_F$. Suppose further that the oracle query
returns the conditional fragment states together with the identification of
record labels with the simulator's basis labels. Then a single oracle query
followed by classical post-processing suffices to construct a valid
decomposition $G_t = U_t P_t U_{t-1}^\dagger$ (absorbing $B_t$), at classical
cost
\begin{equation}
O\!\left( d^{2k} \cdot \poly(D_F) \right),
\label{eq:darw_cost}
\end{equation}
independent of the number of subsystems $n$, the circuit depth $T$, and the
accumulated permutation $P$. In particular, for fixed gate support $k$ and
bounded fragment dimension $D_F$ the cost per gate is $O(1)$, and the
intractable test of Eq.~\eqref{eq:degeneracy_check} is bypassed entirely.
\end{proposition}

\begin{proof}
The oracle returns the $d^{k}$ conditional fragment states $\{\rho^{(i)}_F\}$ on
the gate support. The simulator performs pairwise comparisons: for each of the
$\binom{d^{k}}{2} = O(d^{2k})$ pairs $(i,j)$ it tests whether
$\rho^{(i)}_F \perp \rho^{(j)}_F$ or $\rho^{(i)}_F = \rho^{(j)}_F$; under SBS
these are the only two possibilities. Each test is a fixed linear-algebra
computation on $D_F \times D_F$ matrices, for instance computing
$\tr[\rho^{(i)}_F \rho^{(j)}_F]$ and $\|\rho^{(i)}_F - \rho^{(j)}_F\|$, at cost
$\poly(D_F)$; the comparison cost is governed by the \emph{fragment} dimension,
not by $d^{k}$. This partitions the labels into record classes at total cost
$O(d^{2k}\poly(D_F))$. The labeled correspondence assumed in the statement is
essential: without it the comparisons yield only an \emph{unlabeled} partition of
abstract record indices, which cannot be pulled back to the simulator's basis
states in the $U_{t-1}$ frame and therefore fixes nothing. This is the
label/state gap of Sec.~\ref{sec:labels} in computational form.

Two cases arise for each class. If a class is a singleton, all its fragment
states being orthogonal to the rest, the environment distinguishes that label
and the basis is fixed on it. This includes in particular accidental spectral
degeneracies that the interaction Hamiltonian happens to resolve; the simulator
adopts the environment's basis, which is valid because the environmental record
is itself diagonal in it. If a class contains several labels with coincident
fragment states, the environment certifies that $H_{SE}$ does not resolve that
sector; since coincident records under Darwinian monitoring of the concordant
basis occur only within degenerate sectors of $\rho^{t-1}$, the corresponding
weights are equal and \emph{any} basis choice within the class yields the
correct output distribution (Sec.~\ref{sec:simulation_requirement}).

It remains to exhibit $U_t$. The record classes, pulled back through the labeled
correspondence, give the simulator a partition of the $d^k$ basis labels on the
support into environmentally distinguished singletons and unresolved sectors. On
the singletons the basis vectors are fixed by the environment's own
diagonalizing basis. On an unresolved sector the simulator may adopt any
orthonormal product basis of that sector, and the choice does not affect the
output distribution. This determines a valid target basis on the support
directly, without search; $U_t$ is then read off, and $P_t$ is obtained by
applying $G_t$ to the resulting projectors and recording the induced relabeling.
Every operation is confined to the $d^k$-dimensional support. At no point is the
accumulated permutation $P$ evaluated, and no quantity in the procedure depends
on $n$ or $T$.
\end{proof}

From the perspective of the simulation algorithm the SBS record therefore acts
as a basis-selection oracle. While the reduced density matrix $\rho^{t-1}$ alone
cannot supply the information (recovering it from the circuit history requires
the intractable test), the environmental record encodes it directly. Formally,
the Darwinian mechanism defines a map
\begin{equation}
(\rho^{t-1}, G_t, H_{SE}) \longmapsto \mathcal{D}_t ,
\label{eq:oracle_map}
\end{equation}
where $\mathcal{D}_t$ is the record-class partition, and hence the unresolved
degenerate sector, of $\rho^{t-1}$ on the support of $G_t$. This map resolves
precisely the ambiguity that obstructs simulation for general gate sets.

\subsection{Basis objectivity as a computational resource}
\label{sec:basis_objectivity}

\begin{definition}[Basis objectivity]
\label{def:basis_objectivity}
A local product basis $\mathcal{B}$ of a concordant state $\rho$ is
\emph{objective} if there exist redundant, independently accessible physical
records---readable by multiple observers without access to, or disturbance of,
the system---that single out $\mathcal{B}$, up to arbitrary rotations within
sectors of exactly equal weight, which do not affect any observable statistics.
\emph{Basis non-objectivity} is present when multiple inequivalent basis
descriptions remain admissible and no record accessible to observers selects
among them.
\end{definition}

The definition is deliberately mechanism-independent: it refers only to the
existence and accessibility of records. Quantum Darwinism---specifically, an
interaction producing SBS---is the physical mechanism by which such records
arise, but the definition does not presuppose it. In the terminology of
Sec.~\ref{sec:labels} basis objectivity is \emph{state} objectivity. Generic SBS
guarantees only the weaker label objectivity, so basis objectivity is a
genuinely stronger condition than SBS alone, and it is the computationally
operative one, since the simulator requires $U_t$ as an explicit matrix.

Basis objectivity is also stronger than the mere existence of a preferred basis.
A basis may be selected by decoherence at the level of the system while
remaining ambiguous to external observers; objectivity requires that the basis
information be actively broadcast in a redundant and accessible form.

The complementary notion, basis non-objectivity, has a direct computational
reading. When it is present the degenerate sector cannot be identified from
environmental records and the simulator must evaluate the intractable test from
the circuit history. It has no analogue in classical probability theory: by
Corollary~\ref{cor:classical_basis} classical pure states admit a unique
decomposition and degeneracy ambiguity cannot arise. Unlike entanglement, basis
non-objectivity is present even in zero-discord states; unlike discord, it is
not a measure of correlations but of the accessibility of basis information. It
represents an axis of quantum advantage distinct from the resources identified
in prior work on pure- and mixed-state computation. This motivates the
hierarchy
\begin{align}
\text{basis objectivity} \ &\longrightarrow \ \text{efficient basis tracking}
\nonumber\\ &\longrightarrow \ \text{simulability},
\nonumber\\[2pt]
\text{basis non-objectivity} \ &\longrightarrow \ \text{sector unidentifiable}
\nonumber\\ &\longrightarrow \ \text{obstruction},
\label{eq:hierarchy_resource}
\end{align}
which Sec.~\ref{sec:bam} makes quantitative.

\subsection{Main conjecture}
\label{sec:conjecture}

The observations above motivate the following conjecture. It is a statement
about two computational models at once, and we flag this explicitly: direction
(i) concerns the oracle model of Sec.~\ref{sec:oraclemodel}, in which the
environmental record is physically given, while direction (ii) concerns the
standard model, in which it is not. The content of the conjecture is precisely
that the two models are separated by the physical presence or absence of the
record.

\begin{conjecture}
\label{conj:main}
Assume $\PP \neq \NP$. Efficient classical simulation of concordant quantum
dynamics is possible if and only if the degenerate-sector structure of the state
on each gate support is efficiently identifiable at every step of the evolution.
Specifically:
\begin{enumerate}
\item[(i)] (``if'', oracle model) When basis objectivity is present---i.e.\ when
SBS records of the concordant basis exist and are accessible---the
identification is efficient and Eastin's simulation extends to arbitrary fixed
gate support. This direction is established, in the oracle model of
Sec.~\ref{sec:oraclemodel}, by Proposition~\ref{prop:darwinres}.
\item[(ii)] (``only if'', standard model) In the absence of any efficiently
accessible record of the degeneracy structure, no classical algorithm can
simulate general concordant circuits efficiently. Equivalently, basis
non-objectivity is an obstruction that cannot be circumvented. A resolution of
this direction would in particular resolve Eastin's open problem
\cite{eastin2010simulating}; it remains open.
\end{enumerate}
\end{conjecture}

\subsection{Evidence and obstacles for the converse}
\label{sec:converse}

Since direction (ii) carries the main conceptual weight, we record what is
known, what a counterexample would require, and why we believe it.

\emph{What is established.} The specific \emph{route} through the degeneracy
test is provably hard: deciding Eq.~\eqref{eq:degeneracy_check} for gates of
support three or more is NP-hard (Remark~\ref{rem:hardness}). What is not
established is that every simulation strategy must pass through this test.
Conceivably, an algorithm could sample from the correct output distribution
without ever identifying a diagonal basis---the situation familiar from weak
simulation, which is sometimes strictly easier than strong simulation, and for
which a variety of techniques exist that never reconstruct the state at all
\cite{nest2010simulatingquantumcomputersprobabilistic}. Cable and Browne
\cite{Cable_2015} do extend efficient simulation to a larger class by exploiting
additional structure, so direction (ii) is to be read as a statement about the
general problem and not about every circuit with degeneracy.

\emph{What a counterexample would look like.} A refutation of (ii) would be an
efficient classical sampler for concordant circuits with three-qubit gates
working for worst-case instances, including those whose accumulated permutations
encode hard Boolean formulas. Such a sampler would have to extract the output
distribution from the circuit description while remaining insensitive to the
degeneracy structure that determines which gate decompositions are valid. We
regard this as implausible for the following reason: the output distribution
itself depends on $P_t$, and by
Proposition~\ref{prop:concordance_preserving} the decomposition of $G_t$ into
$U_t P_t B_t U_{t-1}^\dagger$ is not unique precisely when degeneracy is
present. Distinct admissible decompositions agree on the output distribution
only \emph{because} the ambiguity is confined to equal-weight sectors, so any
correct sampler implicitly certifies which sectors those are, which is the hard
problem. Turning this heuristic into a proof would require showing that sampling
itself embeds a hard decision problem---a step that would place concordant
sampling alongside other hardness-of-sampling results, and which we leave as the
central open problem.

\emph{The gap between the models.} Direction (i) holds in the oracle model,
where the environment's records are physically given; direction (ii) is a
statement about the standard model, where they are not. The conjecture as a
whole therefore asserts that the \emph{physical} presence or absence of
Darwinian records tracks the \emph{computational} boundary exactly. It is this
correspondence---not either direction in isolation---that we propose as the
content of basis objectivity as a resource.

\section{Quantifying the ambiguity: the Basis Ambiguity Measure}
\label{sec:bam}

Section~\ref{sec:darwinism} identified the resource whose absence the
obstruction is, and showed what supplying it accomplishes. That analysis was
qualitative on one count: it did not say \emph{how much} basis freedom a given
concordant state possesses. We now assign the freedom a number.

\subsection{Degenerate sectors on the gate support}
\label{sec:sectors}

Let $G_t$ have support $b$ consisting of $k = |b|$ subsystems of local dimension
$d$, so $d_b = d^{\,k}$, and let $a$ denote the remainder of the register.
Because $G_t$ acts only on $b$, the basis on $a$ is untouched at this step and
the entire question of which bases remain admissible is settled by the structure
of $\rho$ relative to the cut $a|b$. At the next step the support changes and
the cut is redrawn; the promise of the computation is that the register is
concordant throughout.

Writing $\vec{x}$ for the product basis label on $b$ and
$\Pi^{(b)}_{\vec{x}} = \bigotimes_{j\in b}\Pi^{(j)}_{x_j}$, concordance gives
Eq.~\eqref{eq:conditional_form}. The unnormalized conditional state
$\tilde{\rho}^{(a)}_{\vec{x}}$ plays the role an eigenvalue plays when $b$ is
the whole register. Eastin calls two labels \emph{degenerate on part $b$} when
these operators coincide \cite{eastin2010simulating}; Cable and Browne call
$\tilde{\rho}^{(a)}_{\vec{x}}$ the subsystem operator-valued eigenvalue
\cite{Cable_2015}. Grouping labels by their conditional state partitions the
$d_b$ labels into \emph{degenerate sectors}
\begin{equation}
S_i = \bigl\{\vec{x} : \tilde{\rho}^{(a)}_{\vec{x}} = \sigma_i\bigr\},
\quad
\Pi_i = \sum_{\vec{x}\in S_i}\Pi^{(b)}_{\vec{x}},
\quad
m_i = |S_i| ,
\label{eq:sectors}
\end{equation}
where the $\sigma_i$ are the distinct conditional states and
$\sum_i m_i = d_b$. The $\Pi_i$ are the full-rank subsystem eigenprojectors of
Ref.~\cite{Cable_2015} and the decomposition is unique. A block-diagonal $B_t$
in Eq.~\eqref{eq:gate_degenerate} is admissible precisely when
$[B_t,\Pi_i]=0$ for every $i$, so the sector sizes $m_i$ control exactly how
much freedom the state grants the gate.

\subsection{Definition of \texorpdfstring{$\beta$}{beta}}
\label{sec:beta_def}

Removing physically irrelevant phases, the unitaries acting within the $i$-th
sector and leaving $\rho$ invariant form the flag manifold
$U(m_i)/U(1)^{m_i}$, of real dimension $m_i(m_i-1)$. The set of orthonormal
bases of the support that diagonalize $\rho$---equivalently, by
Proposition~\ref{prop:concordance_preserving}, the set of admissible
block-diagonal gate factors $B_t$---is therefore the product manifold
\begin{equation}
\mathcal{V}(\rho,b) = \prod_i \frac{U(m_i)}{U(1)^{m_i}} ,
\label{eq:diag_manifold}
\end{equation}
whose real dimension is $\sum_i m_i(m_i-1)$. The maximum, $d_b(d_b-1)$, is
attained by a single sector spanning the support.

$\mathcal{V}$ is the orthogonal freedom of the degenerate regime of
Sec.~\ref{sec:entangled_mixture}. Its elements that are again local \emph{product}
bases form a subset, and generally a proper one, since a rotation inside a
sector whose labels are far apart in Hamming distance produces entangled
vectors. Counting $\mathcal{V}$ rather than that subset is the right choice
because $\mathcal{V}$ is the freedom the state grants the \emph{gate}, and by
Theorem~\ref{thm:beta_pos} a permutation always converts it into product form.

\begin{definition}[Basis Ambiguity Measure]
\label{def:BAM}
The \emph{Basis Ambiguity Measure} of a concordant state $\rho$ with respect to
a gate support $b$ of dimension $d_b$ is
\begin{equation}
\beta(\rho,b) \;=\; \frac{\mathcal{A}(\rho,b)}{d_b(d_b-1)},
\qquad
\mathcal{A}(\rho,b) \;=\; \sum_i m_i(m_i-1),
\label{eq:BAM}
\end{equation}
where $\{m_i\}$ are the sizes of the degenerate sectors~\eqref{eq:sectors}. By
construction $\beta \in [0,1]$, with $\beta = 0$ exactly when every sector is a
singleton and $\beta = 1$ exactly when a single sector spans the support.
\end{definition}

We write $\beta$ for $\beta(\rho,b)$ wherever the state and support are clear
from context. The quantity $\mathcal{A}$ is the real dimension of the manifold
of concordance-preserving gate factors the state admits, and $\beta$ its
normalized version. Note that $\beta$ is a property of the pair $(\rho,b)$ and
not of $\rho$ alone: the same state can carry different ambiguity relative to
different gate supports, as expected, since the gate is what will
exploit this freedom. Section~\ref{sec:hardy_excess} shows that the normalization
$d_b(d_b-1)$ is Hardy's excess on $b$, so that $\beta$ is the fraction of quantum
theory's non-classical degrees of freedom on the support that the state leaves
undetermined.

\subsection{What \texorpdfstring{$\beta$}{beta} characterizes}
\label{sec:equivalence}

The following two results establish a sharp equivalence.

\begin{theorem}[$\beta=0$ is safe]
\label{thm:beta_zero}
If $\beta(\rho,b) = 0$ then every admissible concordance-preserving gate on $b$
leaves the LBF successful, and the returned local basis is the correct $U_t$.
\end{theorem}

\begin{proof}
$\beta = 0$ makes every sector a singleton, so each $\Pi_i$ is a rank-one
product projector and $[B_t,\Pi_i]=0$ for all $i$ forces $B_t$ to be diagonal in
the basis $\{\Pi^{(b)}_{\vec{x}}\}$. A diagonal $B_t$ commutes with every
computational basis projector $X_k$ on the support, so
\begin{equation}
G_t X_k G_t^\dagger
= U_t P_t B_t X_k B_t^\dagger P_t^\dagger U_t^\dagger
= U_t \bigl(P_t X_k P_t^\dagger\bigr) U_t^\dagger .
\end{equation}
Every transformed projector is therefore a local-unitary conjugate of a
computational basis projector. Each such operator commutes with a rank-one
projector on subsystem $j$ precisely when that projector belongs to the basis of
$U_t$ on $j$; hence every retained constraint names the same basis and no two
can be incompatible.
\end{proof}

The converse is the more useful direction, and it is fully constructive:
whenever the state carries any degeneracy at all on the support, one can
\emph{always} write down a gate that preserves concordance and yet defeats the
LBF. The construction requires nothing beyond the identity of two labels in a
common sector.

\begin{theorem}[$\beta>0$ always admits a defeating gate]
\label{thm:beta_pos}
Let $k \geq 2$. If $\beta(\rho,b) > 0$ then there exists an admissible
concordance-preserving gate on $b$ that defeats the LBF. Explicitly, it may be
taken to be $G = PB$, where
\begin{enumerate}
\item $\vec{x} \neq \vec{y}$ are any two labels lying in a common degenerate
sector;
\item $P$ is any permutation of the product basis labels for which $P(\vec{x})$
and $P(\vec{y})$ differ in exactly one subsystem index;
\item $B$ is the rotation by $\theta \notin \tfrac{\pi}{2}\mathbb{Z}$ inside
$\mathrm{span}\{\ket{\vec{x}},\ket{\vec{y}}\}$, acting as the identity on every
other label.
\end{enumerate}
\end{theorem}

\begin{proof}
Since $\beta > 0$ some sector $S_i$ has $m_i \geq 2$; pick distinct
$\vec{x},\vec{y}\in S_i$. Choose a permutation $P$ of the product basis labels
such that $P(\vec{x})$ and $P(\vec{y})$ differ in exactly one subsystem index
$j$; such a $P$ exists because the permutation group acts transitively on
ordered pairs of distinct labels. Let $B$ act as the rotation
\begin{align}
B\ket{\vec{x}} &= \cos\theta\,\ket{\vec{x}} + \sin\theta\,\ket{\vec{y}},
\nonumber\\
B\ket{\vec{y}} &= -\sin\theta\,\ket{\vec{x}} + \cos\theta\,\ket{\vec{y}},
\end{align}
with $\theta \notin \tfrac{\pi}{2}\mathbb{Z}$, and as the identity on all other
labels. Because $\vec{x}$ and $\vec{y}$ lie in a common sector, $B$ commutes
with $\rho$, so $B$ is an admissible block-diagonal factor and $G = PB$ is
concordance-preserving, with $\rho_t = P\rho P^\dagger$ diagonal in the
computational product basis. Take $U_{t-1}=\id$.

Now apply the LBF to $G$. For $X = \proj{\vec{x}}$,
\begin{equation}
G X G^\dagger = \proj{\varphi},
\quad
\ket{\varphi} = \cos\theta\,\ket{P(\vec{x})} + \sin\theta\,\ket{P(\vec{y})} ,
\end{equation}
and since $P(\vec{x})$ and $P(\vec{y})$ differ in the single index $j$,
$\ket{\varphi}$ is a product vector whose $j$-th factor is
$\cos\theta\ket{P(\vec{x})_j} + \sin\theta\ket{P(\vec{y})_j}$ and whose other
factors are computational. This $X$ is retained and forces the rotated basis on
subsystem $j$. Because $d_b \geq 3$ there is a label
$\vec{z} \notin \{\vec{x},\vec{y}\}$, for which
$G\proj{\vec{z}}G^\dagger = \proj{P(\vec{z})}$, a computational product
projector forcing the computational basis on subsystem $j$. Since
$\theta \notin \tfrac{\pi}{2}\mathbb{Z}$ the two demanded bases on $j$ do not
commute, so the constraints are incompatible and the LBF halts.
\end{proof}

\begin{corollary}[Sharp threshold]
\label{cor:threshold}
For $k \geq 2$, an admissible concordance-preserving gate defeating the LBF
exists if and only if $\beta(\rho,b) > 0$.
\end{corollary}

The two theorems place the threshold exactly at the onset of sector degeneracy,
and Theorem~\ref{thm:beta_zero} is the mixed-state counterpart of Eastin's
result establishing efficient simulation for generic, non-degenerate product
inputs at \emph{arbitrary} gate support \cite{eastin2010simulating}. In the
contrapositive direction, Corollary~\ref{cor:threshold} says that $\beta > 0$ is
\emph{necessary} for either known route to concordant simulation to encounter
difficulty at that step.

\begin{remark}[What is not claimed]
\label{rem:not_claimed}
Corollary~\ref{cor:threshold} concerns the existence of a gate defeating a
particular algorithm. It does not assert that $\beta > 0$ makes classical
simulation hard. The free directions inside a sector carry equal weight, so a
simulator that knows the sector structure may choose arbitrarily among the
admissible bases and still reproduce the correct output statistics
(Sec.~\ref{sec:simulation_requirement}); and Cable and Browne show that some
ambiguities are resolved by applying the LBF to extended gate sequences, with
efficiency preserved when the ambiguous stretch is logarithmic in the circuit
size \cite{Cable_2015}. Following Eastin's framing, $\beta$ is offered as a
necessary resource, in the same logical position that entanglement occupies for
pure-state computation, where the Gottesman--Knill theorem
\cite{gottesman1998heisenberg} shows that necessity does not imply sufficiency.
\end{remark}

\subsection{Worked examples}
\label{sec:worked_examples}

We take two qubits with $b$ the whole register, so that $a$ is empty, the
conditional states reduce to probabilities, $d_b = 4$ and
$\beta = \mathcal{A}/12$. A four-qubit example with a three-qubit support
follows.

\paragraph{Case 1: no degeneracy.}
$\rho = \mathrm{diag}(0.4,0.3,0.2,0.1)$. Four singleton sectors,
$\mathcal{A}=0$, $\beta = 0$. Every admissible $B_t$ is a diagonal phase matrix;
by Theorem~\ref{thm:beta_zero} no gate can defeat the LBF, whatever the support
size.

\paragraph{Case 2: one degenerate pair.}
$\rho = \mathrm{diag}(0.4,0.3,0.15,0.15)$. Sectors $\{00\},\{01\},\{10,11\}$, so
$\mathcal{A} = 2$ and $\beta = 1/6$. The construction of
Theorem~\ref{thm:beta_pos} applies directly with $P = \id$, since $10$ and $11$
already differ in one index: the gate rotating qubit~2 inside
$\mathrm{span}\{\ket{10},\ket{11}\}$ makes $\proj{00}$ demand the computational
basis on qubit~2 while $\proj{10}$ demands the rotated one. This is the
controlled-unitary family of Ref.~\cite{Cable_2015} and the two-qubit instance
worked out in Appendix~\ref{app:lbf_failure}.

\paragraph{Case 3: degenerate pairs at Hamming distance two.}
$\rho = \mathrm{diag}(0.3,0.2,0.2,0.3)$. Sectors $\{00,11\}$ and $\{01,10\}$, so
$\mathcal{A}=4$ and $\beta = 1/3$. Here the labels within each sector differ in
\emph{both} indices, so a rotation inside a sector produces entangled vectors,
which the LBF discards; the naive gate therefore does not defeat it. The
permutation in Theorem~\ref{thm:beta_pos} is what supplies the defeating gate:
choosing $P$ to exchange the labels $01$ and $11$ sends
$\cos\theta\ket{00}+\sin\theta\ket{11}$ to
$\cos\theta\ket{00}+\sin\theta\ket{01}$, a product vector, and the
incompatibility follows. The example shows why the theorem is stated over all
admissible gates rather than over a naive family. The reason is the entanglement
criterion of Proposition~\ref{prop:entangled_mixture}: at Hamming distance two
the alternative eigenbases of a sector consist of entangled vectors, which the
LBF discards, so the freedom is real but invisible until a permutation brings it
into product form. This is also why Case~3 behaves differently from Case~2 at
the same value of $\beta$.

The same value $\beta = 1/3$ arises for
$\rho = \mathrm{diag}(0.3,0.3,0.2,0.2)$, whose sectors $\{00,01\}$ and
$\{10,11\}$ also have sizes $(2,2)$. There the labels within each sector differ
in a single index, so the construction applies with $P = \id$ and no permutation
is needed. The two states show that $\beta$ fixes the amount of freedom but not
how it is distributed over the labels.

\paragraph{Case 4: maximal ambiguity.}
$\rho = \id/4$. A single sector spanning all four labels, $\mathcal{A}=12$,
$\beta = 1$. Every unitary on the support is an admissible $B_t$. In particular
the gate carrying the computational basis to the Bell basis, which Cable and
Browne exhibit as a special case \cite{Cable_2015}, is admissible here and
defeats the LBF, as does the elementary construction of
Theorem~\ref{thm:beta_pos}. Maximal degeneracy is the regime in which the state
constrains the gate least, and therefore the regime in which a defeating gate is
easiest to find.

\paragraph{Case 5: four qubits, three-qubit support.}
Let the register be four qubits with $b = \{1,2,3\}$ carrying the gate and
$a=\{4\}$ a spectator, so $k=3$, $d_b = 8$, and $\beta = \mathcal{A}/56$. Assign
to the eight labels $\vec{x}\in\{0,1\}^{3}$ conditional states on qubit~4 that
are pairwise distinct except for $\vec{x} = 110$ and $\vec{x} = 111$, which
coincide. The sectors are six singletons together with $S = \{110,111\}$, so
$\mathcal{A}=2$ and $\beta = 1/28$. The gate applying a rotation to qubit~3
within $S$, embedded as $G\otimes\id$ on the spectator, is unitary, commutes
with $\rho$ in its $B$ factor, preserves concordance, and drives the LBF to
incompatible solutions. Two features of the general setting are visible. The
spectator qubit never enters the LBF's computation, which is carried out
entirely on the eight-dimensional support; it enters only in determining,
through its conditional states, which $B_t$ are admissible. And the sector
structure is a property of the cut $a|b$, so the same four-qubit state carries a
different $\beta$ when the gate acts on a different triple.

\subsection{Cost is not monotonic in \texorpdfstring{$\beta$}{beta}}
\label{sec:nonmono}

Corollary~\ref{cor:threshold} locates the threshold exactly, but says nothing
about how \emph{often} the LBF fails once past it, and the two questions have
different answers. Sampling admissible gates $G = U_t P_t B_t$
uniformly (Haar-random local $U_t$, uniform $P_t$, Haar-random $B_t$ on each
sector) for the two-qubit states of Sec.~\ref{sec:worked_examples} gives a
failure frequency that first rises and then collapses
(Table~\ref{tab:nonmono}; the procedure is given in
Appendix~\ref{app:sampling}). For $\mathrm{diag}(0.4,0.3,0.15,0.15)$, with a
single degenerate pair and $\beta = 1/6$, a clear majority of sampled gates
defeat the algorithm; for the two states with $\beta = 1/3$ the proportion is
higher still. But for $\mathrm{diag}(0.4,0.2,0.2,0.2)$, whose single sector has
three labels and $\beta = 1/2$, and for $\id/4$ with $\beta = 1$, essentially no
sampled gate defeats it at all.

\begin{table}[t]
\caption{Failure frequency of the Local Basis Finder under uniform sampling of
admissible concordance-preserving gates on a two-qubit support
(Appendix~\ref{app:sampling}), as a function of the sector structure. The
threshold is at $\beta>0$ (Corollary~\ref{cor:threshold}), but the frequency
with which a randomly drawn gate realizes the threat is not monotone in $\beta$.
The last column is qualitative by design: it is a property of one algorithm
under one distribution over gates, not of the state.}
\label{tab:nonmono}
\begin{ruledtabular}
\begin{tabular}{lccl}
state (diagonal) & sectors $\{m_i\}$ & $\beta$ & LBF defeated \\
\colrule
$(0.4,0.3,0.2,0.1)$      & $1,1,1,1$ & $0$   & never (Thm.~\ref{thm:beta_zero})\\
$(0.4,0.3,0.15,0.15)$    & $1,1,2$   & $1/6$ & majority of samples \\
$(0.3,0.2,0.2,0.3)$      & $2,2$     & $1/3$ & large majority \\
$(0.3,0.3,0.2,0.2)$      & $2,2$     & $1/3$ & large majority \\
$(0.4,0.2,0.2,0.2)$      & $1,3$     & $1/2$ & essentially never \\
$\id/4$                  & $4$       & $1$   & essentially never \\
\end{tabular}
\end{ruledtabular}
\end{table}
The mechanism is the property of the algorithm rather than the state.
 On a sector of size greater than two a Haar-random $B_t$
generically maps rank-one product projectors to entangled projectors, and the
LBF discards these rather than treating them as constraints; with nothing
retained there is nothing to conflict. The typical gate is therefore harmless
precisely where the state is most degenerate, even though by
Theorem~\ref{thm:beta_pos} a structured gate defeating the algorithm exists in
every one of these cases. For $\id/4$, for instance, every unitary on the
support is an admissible $B_t$, so the gate carrying the computational basis to
the Bell basis is available and it defeats the LBF, as does the elementary
two-label rotation of Theorem~\ref{thm:beta_pos}. What the sampling measures is
how easy such gates are to hit at random, not whether they exist.

Thus $\beta$ tells us how much freedom the state gives to the gate; it does not
tell us how often that freedom is used. This behavior should not be read as a
property of the resource. What is measured is the sensitivity of one particular
algorithm to one particular distribution over gates, and the LBF is moreover
known to be incomplete: it may halt on gates for which other strategies succeed,
and the characterization of such gates is left open in
Ref.~\cite{Cable_2015}. We therefore record the behavior of the simulation cost
as a function of $\beta$ as an open problem rather than as a result.

\begin{remark}[Analogy with Vidal's theorem, and its limits]
\label{rem:vidal}
Vidal \cite{PhysRevLett.91.147902} showed that the classical simulation cost of
a pure-state circuit is $\poly(n)\cdot\poly(\chi)$ with $\chi$ the Schmidt rank,
so that logarithmic entanglement $E_\chi = \log\chi$ still permits efficient
simulation. In concordant computation $E_\chi = 0$ and discord vanishes, yet the
manifold of admissible gate factors has dimension
$\mathcal{A}(\rho,b) = \beta\, d_b(d_b-1)$, which is the natural candidate for
the graded parameter. A bound of the form
$\mathrm{Cost} = f(n,\mathcal{A})$, with $f$ polynomial in $n$ at fixed
$\mathcal{A}$, would be the mixed-state counterpart of Vidal's theorem. We
emphasize that we have not established such a bound, and that
Sec.~\ref{sec:nonmono} shows the naive expectation of monotone growth is not
supported by the one algorithm currently available.
\end{remark}

\section{Basis ambiguity as a formal resource}
\label{sec:resource}

We now place $\beta$ within the Chitambar--Gour framework
\cite{chitambar2019quantum}, identifying the free states, the operations under
which $\beta$ is invariant, and those under which it is not.

\subsection{Free states and free operations}

\paragraph{Free states.} The free states are those whose sectors are all
singletons,
\begin{equation}
\mathcal{F}_b = \{\rho \mid \beta(\rho,b) = 0\} .
\label{eq:free_states}
\end{equation}
By Theorem~\ref{thm:beta_zero} these are exactly the states on which no
admissible gate can defeat basis tracking.

\paragraph{Free operations.} The free operations are those preserving the sector
structure,
\begin{align}
\mathcal{O}_b = \Bigl\{\, W_a &\otimes
  \bigl(\textstyle\bigotimes_{j\in b} V_j\bigr)\cdot \sigma_b \ \Big|\
  W_a \in U(d_a), \nonumber\\
  &\ V_j \in U(d),\ \sigma_b \text{ a product of relabelings} \Bigr\},
\label{eq:free_ops}
\end{align}
together with any $B_t$ commuting with the state.

\subsection{Invariance and faithfulness}

\begin{theorem}[Invariance and faithfulness]
\label{thm:beta_invariance}
Fix a cut $a|b$. Then $\beta(\rho,b)$ is invariant under every
concordance-preserving unitary supported entirely within $b$, and under every
unitary supported entirely within $a$. Moreover $\beta(\rho,b) = 0$ if and only
if $\rho \in \mathcal{F}_b$.
\end{theorem}

\begin{proof}
Faithfulness is Eq.~\eqref{eq:free_states}. A unitary supported within $a$ acts
as $\tilde\rho^{(a)}_{\vec{x}} \mapsto W_a \tilde\rho^{(a)}_{\vec{x}}
W_a^\dagger$ with the same $W_a$ for every $\vec{x}$. Since conjugation by a
fixed unitary is injective on operators, two labels share a conditional state
after the map exactly when they did before, so the sectors and their sizes $m_i$
are unchanged. A unitary supported within $b$ has the form $U_b P_b B_b$ with
$U_b = \bigotimes_{j \in b} V_j$ local. The factor $U_b$ replaces each
$\Pi^{(j)}_{x_j}$ by
$V_j \Pi^{(j)}_{x_j} V_j^\dagger$ without altering which conditional state
belongs to which label; the permutation $P_b$ relabels the sectors without
changing their sizes; and $B_b$ commutes with $\rho$ and so leaves it unchanged.
In both cases $\sum_i m_i(m_i-1)$ is preserved.
\end{proof}

\begin{remark}[The cut is part of the resource]
\label{rem:non_invariance}
A gate straddling the cut, acting on some subsystems of $a$ and some of $b$, is
not covered by Theorem~\ref{thm:beta_invariance}, and it can change
$\beta(\rho,b)$ in either direction. Such a gate rearranges which weight is
attached to which label, and so can make two conditional states coincide that
were previously distinct, or separate two that agreed. Only the value on the
full register is invariant under every concordance-preserving unitary, since
there the conditional states are the eigenvalues themselves and no unitary
changes the spectrum.

This is not a defect of the measure but the normal behavior of any quantity
defined relative to a bipartition. The Schmidt rank $\chi$ of Vidal's
construction behaves the same way: it is fixed for a given cut and is altered by
gates acting across it. In a concordant circuit the cut at step $t$ is defined
by the gate itself, since $b_t$ is the gate's support. The value that matters at
that step, $\beta(\rho_{t-1}, b_t)$, is therefore invariant under the very gate
that would exploit it, and the resource available at the next step is fixed by
what earlier gates on other supports have done.

Non-unitary operations change $\beta$ on any support. A measurement collapses
the sector structure, and classical mixing can cause distinct conditional states
to coincide, merging sectors and raising $\beta$. Whether a genuine monotone
exists for a class of operations that includes the non-unitary ones is open, and
is the natural next step in the resource-theoretic development. Convexity, the
requirement that mixing free states cannot create a resourceful one, fails for
the same reason, but mixing is not among the free operations
\eqref{eq:free_ops}, and the failure is a physically correct reflection of the
fact that classical mixing generates ambiguity.
\end{remark}

\subsection{The boundary between probabilistic theories}
\label{sec:gpt}

For a generalized probabilistic theory $\mathcal{G}$ with transformation group
$\Gamma_{\mathcal{G}}$ \cite{PhysRevA.75.032304}, let
$\mathcal{V}_{\mathcal{G}}(\rho,b)$ be the set of valid diagonalizing product
bases and extend the measure as its normalized dimension, in parallel with
Eq.~\eqref{eq:diag_manifold}. If $\Gamma_{\mathcal{G}}$ is a discrete group, as
for the polygon GPTs, whose transformation groups are finite
\cite{Janotta_2011,PhysRevA.75.032304}---then $\mathcal{V}_{\mathcal{G}}$ is a
finite set of isolated points, of real topological dimension zero, so
$\beta_{\mathcal{G}} = 0$ for every state. Such theories contain no continuous
basis ambiguity at all: there is no manifold of admissible bases to search, and
the obstruction identified here is simply absent.

This is a statement about the structure of the ambiguity and not by itself a
claim of efficient simulability, since a finite set of admissible bases may
still be large. The parameter $\beta$ thus registers that continuous basis
ambiguity arises only when the underlying theory permits a continuous
transformation group, and in this sense is a computational witness to Hardy's
fifth axiom.

\section{Darwinism of the Basis and axiomatic collapse}
\label{sec:dob}

\subsection{Darwinism of the Basis}

Section~\ref{sec:bam} established that $\beta(\rho,b) = 0$ is exactly the
condition under which no admissible gate can defeat basis tracking on $b$. We
now interpret this through the lens of quantum Darwinism, connecting it to Sec.~\ref{sec:darwinism}.

\begin{definition}[Darwinism of the Basis]
\label{def:dob}
A concordant circuit achieves \emph{Darwinism of the Basis} (DOB) when the
diagonalizing basis at every computational step is objectively
available (known, broadcast, and unambiguous) so that any classical simulator
can access it without exponential search.
\end{definition}

DOB is the per-step, computational form of basis objectivity. In
Sec.~\ref{sec:darwinism} basis objectivity was found to require two ingredients
at once: the record classes supplied by an SBS environment, which is a property
of the state, and the interaction map supplied by $H_{SE}$, which is a property
of the dynamics, with neither alone sufficient. DOB is the demand that this
conjunction hold at every gate of the circuit.

By Corollary~\ref{cor:threshold} we then have
\begin{align}
&\text{no admissible gate can defeat basis tracking} \nonumber\\
&\hspace{2.2em}\iff\ \beta(\rho_t,b_t) = 0 \ \text{ for all } t ,
\label{eq:dob_iff}
\end{align}
which is the regime in which efficient simulation at arbitrary gate support is
already guaranteed for generic inputs \cite{eastin2010simulating}. When the
condition fails, Sec.~\ref{sec:darwinism} identifies the physical mechanism that
can restore it, namely an environment in spectrum broadcast form together with the
interaction map, and Proposition~\ref{prop:darwinres} bounds the cost of
extracting it. In this sense DOB is the computational signature of
environmental objectivity: where the environment has fixed the basis, the
simulator inherits it for free.

Conversely, when DOB fails, a gate exists that leaves the basis undetermined by
the transcript.
Equation~\eqref{eq:dob_iff} concerns basis tracking, not classical simulability
as such; the caveats of Remark~\ref{rem:not_claimed} and the non-monotonic
behavior of Sec.~\ref{sec:nonmono} both apply. The claim that basis
non-objectivity obstructs \emph{every} efficient classical strategy is
Conjecture~\ref{conj:main}(ii) and remains open; it is that conjecture, and not
Eq.~\eqref{eq:dob_iff}, that would place the quantum--classical computational
divide exactly at the onset of DOB.

\begin{figure}[t]
\centering
\includegraphics[width=0.95\linewidth]{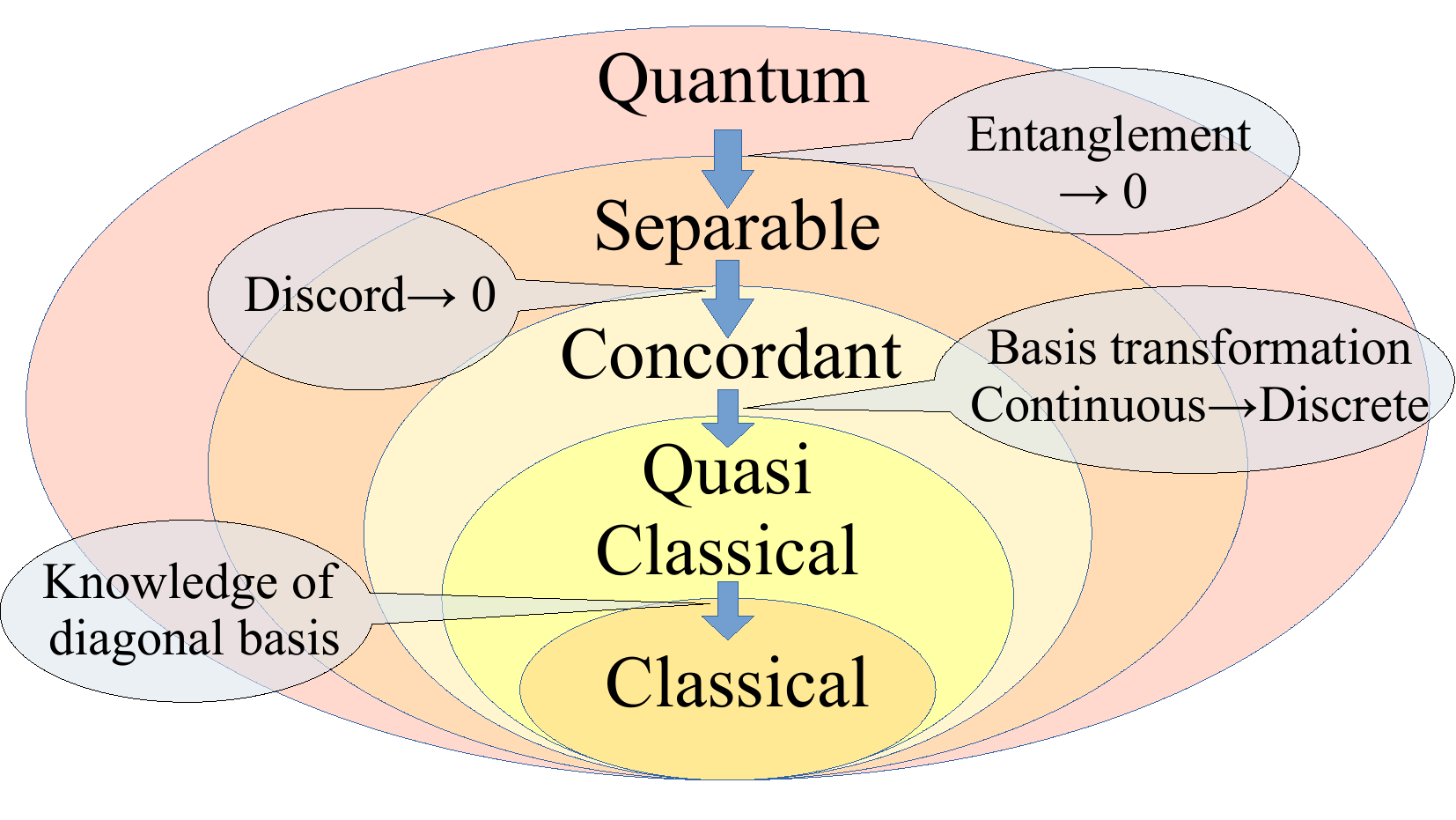}
\caption{Hierarchy of computational models illustrating the transition from
fully quantum to classical computation, annotated by the ambiguity dimension
$\mathcal{A}$ and its normalized measure $\beta$ of Sec.~\ref{sec:bam}. Quantum
computations with an advantage over classical models occupy the topmost regime
($\mathcal{A}>0$, $\beta>0$), enabled by entanglement and discord as genuine
quantum resources. In the concordant sublayer, entanglement and discord vanish
but $\beta>0$ still admits gates that obstruct direct basis tracking. The third
arrow marks the transition to discrete-group theories (quasi-classical layer);
the fourth marks the DOB condition $\beta=0$, the onset of full classical
objectivity. The diagram makes quantitative the resource-theoretic hierarchy of
Eq.~\eqref{eq:hierarchy_resource}.}
\label{fig:hierarchy}
\end{figure}

\subsection{Hardy's excess degrees of freedom and the meaning of \texorpdfstring{$\beta$}{beta}}
\label{sec:hardy_excess}

Hardy characterizes a system by the number $N$ of perfectly distinguishable
states and the number $K$ of real parameters that specify a general state
\cite{hardy2001quantum}. Classical probability theory has $K = N$ and quantum
theory $K = N^2$; the continuity requirement of the fifth axiom is what excludes
$K = N$, and continuous reversibility plays the same role in later
reconstructions \cite{Masanes2011derivation}. Writing
$\rho = V\,\mathrm{diag}(\vec{p})\,V^\dagger$, the $N$ eigenvalues are the
classical part of the description, and the $K - N = N(N-1)$ excess degrees of
freedom span the flag manifold $U(N)/U(1)^N$ that fixes the eigenbasis. Hardy's
excess degrees of freedom are thus exactly the information that says
\emph{which basis} the probabilities belong to.

The orbit of $\rho$ on $b$, $\mathcal{O}_b(\rho) = \{(\id_a\otimes V)\,\rho\,(\id_a\otimes V)^\dagger : V \in U(d_b)\}$,
is the set of states reachable by rotating the support, and its real dimension counts
the directions in which such a rotation actually changes $\rho$. For a single qubit,
$\mathrm{diag}(p,1-p)$ with $p \neq \tfrac12$ has as its orbit the Bloch sphere of radius $|2p-1|$
(dimension $2 = K-N$, so $\beta = 0$), whereas $\id/2$ has a single point (dimension $0$, so $\beta = 1$).

\begin{proposition}[Ambiguity as unfixed excess degrees of freedom]
\label{prop:hardy_excess}
Let $\rho$ be concordant, let $b$ be a gate support of dimension $d_b$ with
degenerate sectors of sizes $\{m_i\}$ and projectors $\{\Pi_i\}$, and let
$K_b - N_b = d_b(d_b-1)$ be the number of Hardy's excess degrees of freedom on $b$. The unitaries on $b$ that
leave $\rho$ invariant are exactly those that commute with every sector
projector,
\begin{equation}
\mathrm{Stab}_b(\rho) = \{V \in U(d_b) : [V,\Pi_i] = 0\ \forall i\}
  \cong \prod_i U(m_i),
\label{eq:stab}
\end{equation}
and the orbit $\mathcal{O}_b(\rho)$ of $\rho$ under unitaries on $b$ satisfies
\begin{align}
\dim_{\mathbb{R}}\mathcal{O}_b(\rho) + \mathcal{A}(\rho,b) &= K_b - N_b ,
\label{eq:hardy_excess}\\
\beta(\rho,b) &= 1 - \frac{\dim_{\mathbb{R}}\mathcal{O}_b(\rho)}{K_b - N_b}.
\nonumber
\end{align}
\end{proposition}

\begin{proof}
Write $\rho = \sum_i \sigma_i \otimes \Pi_i$ with distinct $\sigma_i$, as in
Eq.~\eqref{eq:sectors}. A unitary $V$ on $b$ that commutes with every $\Pi_i$
fixes $\rho$. Conversely, if $V$ fixes $\rho$ it commutes with every operator
$\tr_a[(X\otimes\id_b)\rho] = \sum_i \tr(X\sigma_i)\,\Pi_i$. Because the
$\sigma_i$ are distinct, for each pair $i \neq j$ some $X$ gives
$\tr(X\sigma_i) \neq \tr(X\sigma_j)$, so polynomials in these operators,
together with the identity, reach every $\Pi_i$, and $V$ commutes with each of
them. The commutant of $\{\Pi_i\}$ in $U(d_b)$ is $\prod_i U(m_i)$, of real
dimension $\sum_i m_i^2$, and the orbit--stabilizer theorem gives
$\dim_{\mathbb{R}}\mathcal{O}_b = d_b^2 - \sum_i m_i^2$. Since
$\sum_i m_i = d_b$, $\mathcal{A} = \sum_i m_i^2 - d_b$, and adding the two gives
Eq.~\eqref{eq:hardy_excess}.
\end{proof}

The orbit dimension counts the orientation parameters the state records, since
rotating the basis in those directions changes $\rho$; $\mathcal{A}$ counts the
directions in which the basis can rotate without changing $\rho$, which by
Eq.~\eqref{eq:stab} are exactly the block factors $B_t$ with phases removed. The
normalization of $\beta$ is therefore Hardy's excess count itself, and $\beta$ is the
share of quantum theory's non-classical degrees of freedom on $b$ that the state
leaves undetermined. At $\beta = 0$ the state fixes its basis. At $\beta = 1$ the
orbit is a point: $\rho = \sigma_a \otimes \id_b/d_b$ is invariant under every
gate on $b$, and the step is trivial for any simulator. Only in between can a
gate act nontrivially on the state while exploiting freedom the state does not
fix. In classical probability theory $K_b = N_b$, and both terms of
Eq.~\eqref{eq:hardy_excess} vanish.

\subsection{Axiomatic collapse}
\label{sec:axiomatic_collapse}

The DOB condition and the inter-theory measure of Sec.~\ref{sec:gpt} point to
the same axiomatic boundary.

\begin{theorem}[Axiomatic collapse]
\label{thm:hardy_collapse}
Let $\mathcal{T}$ be a theory satisfying Hardy's axioms with the continuity
requirement of the fifth axiom weakened to bare reversibility. Then:
\begin{enumerate}
\item[(i)] the state space of $\mathcal{T}$ is a simplex, and every state has a
unique decomposition into perfectly distinguishable pure states;
\item[(ii)] no state of $\mathcal{T}$ is a mixture of pure states carrying
correlations that the state itself does not carry; in particular the family
\eqref{eq:phi_pair} does not exist, and a state with no entanglement is not a
mixture of entangled pure states;
\item[(iii)] neither regime of Sec.~\ref{sec:entangled_mixture} occurs, so
$\beta(\rho,b) = 0$ for every state and every support; the diagonalizing product
basis is fixed by the state at every step and there is nothing for a simulator
to identify.
\end{enumerate}
Conversely, if any one of (i)--(iii) fails in a theory satisfying the remaining
axioms, that theory admits continuous reversible transformations between pure
states.
\end{theorem}

\begin{proof}
(i) By Hardy's reconstruction \cite{hardy2001quantum}, a theory satisfying the
remaining axioms with continuity so weakened is classical probability theory,
whose state space is a simplex \cite{PhysRevA.75.032304}; uniqueness of the
decomposition is Corollary~\ref{cor:classical_basis}.

(ii) Immediate from (i): a state with a unique decomposition admits no
alternative ensemble of any kind, entangled or not.

(iii) With $K = N$ a state is specified by its probability vector alone and
carries no orientation parameters, while the reversible transformations are the
permutations of the vertices, a finite group. The stabilizer of every state is
therefore finite, the analogue of $\mathcal{A}$ in Eq.~\eqref{eq:hardy_excess}
vanishes, and $\beta_{\mathcal{T}} = 0$. Equivalently, in the language of
Proposition~\ref{prop:concordance_preserving} there is no continuous family of
admissible $B_t$ to be chosen: $B_t$ rotates pure states within a degenerate
sector while leaving the probability vector unchanged, and demanding a unique
decomposition for every state is exactly the demand that this rotation be
forbidden, so the pure states within each sector are rigidly fixed.

For the converse, note first that (ii) and (iii) each imply (i) in
contrapositive form: a violation of (ii) exhibits a state with two distinct
decompositions, and a violation of (iii) exhibits a state whose diagonalizing
basis is not unique. So it suffices to treat (i). If (i) fails, some state has
more than one decomposition into perfectly distinguishable pure states, and the
state space is therefore not a simplex, since a point of a simplex has a unique
convex decomposition into vertices. Among theories satisfying the remaining
axioms, the non-simplex case is quantum theory, which carries a continuous
reversible group on its pure states \cite{hardy2001quantum}.
\end{proof}

\begin{corollary}[One fact, two failures]
\label{cor:two_failures}
The failure of branch tracking and the failure of basis tracking are the same
structural fact---non-uniqueness of the decomposition of a mixed state---seen at
two levels of refinement. The first appears already at the level of arbitrary
pure-state ensembles, where a concordant state can be rebuilt from entangled
ones (Proposition~\ref{prop:entangled_mixture}); the second requires the
ensemble to be an orthonormal product basis, which needs degeneracy and holds
exactly when $\beta > 0$ (Corollary~\ref{cor:threshold}). Removing the underlying fact
removes both obstructions at once, and by Theorem~\ref{thm:hardy_collapse}
removing it means replacing quantum theory by classical probability theory.
\end{corollary}

\begin{remark}[What the theorem does not assert]
\label{rem:collapse_scope}
The theorem is not the claim that concordant computation is hard because of
Hardy's axiom. It is the weaker and provable claim that the \emph{freedom} on
which both known simulation strategies founder exists only because of that
axiom, and disappears with it; whether the freedom translates into complexity in
the standard model is Conjecture~\ref{conj:main}(ii) and is untouched here.
Note also that the converse asserts the existence of continuous reversible
transformations, not their usefulness: Sec.~\ref{sec:gpt} records that a theory
can have a finite transformation group, hence $\beta_{\mathcal{G}} = 0$, and
still present a large finite set of admissible bases.
\end{remark}

\subsection{Computational content of the collapse}
\label{sec:collapse_computation}

In either theory, simulation is efficient once every gate is known in
decomposed form: sample an initial label, apply the permutations, and track the
frames (Sec.~\ref{sec:simulation_requirement}). The difference lies in whether
the decomposition is given. In classical probability theory the reversible
transformations of a simplex are the permutations of its vertices, so a
reversible gate \emph{is} its own decomposition: it carries labels to labels, and
a randomized reversible computation is simulated by pushing a sampled
configuration through the gates, at a cost linear in the depth for any gate
size and any degeneracy. Eastin already noted that a concordant computation is
only slightly more general than such a randomized classical computation, the
difference being that its eigenbasis may change between steps
\cite{eastin2010simulating}. In quantum theory a gate arrives as a unitary
$G_t = U_tP_tB_tU_{t-1}^\dagger$, and recovering $P_t$ from it is the whole
problem. That problem exists only because of $B_t$. When $\beta = 0$ the block is
a phase and the gate carries labels to the labels of a new local frame
(Theorem~\ref{thm:beta_zero}); when $\beta > 0$ a gate can act on the ensemble
exactly as a permutation does while carrying individual labels to vectors that
are the labels of no common local frame (Theorem~\ref{thm:beta_pos}). By
Proposition~\ref{prop:hardy_excess} the freedom responsible is a manifold of
dimension $\mathcal{A}$, which vanishes in any theory with $K = N$.

Hardness therefore needs two ingredients, and neither suffices alone: the
continuous freedom $\mathcal{A} > 0$, which exists only by continuity, and a
classical structure that hides which labels share a sector, such as the
non-affine permutations of Eastin's reduction. Without the second, all
permutations are affine and, for product inputs, the degeneracy test is
efficient \cite{eastin2010simulating,Cable_2015}; without the first, every gate
is its own decomposition. A network of Toffoli gates acting on noisy bits is
classically trivial, while the same network with block factors inserted on
degenerate sectors is not known to be efficiently simulable in general. The only
difference between the two is a family of transformations that exists because
Hardy's fifth axiom demands continuity.

\begin{remark}[Discrete theories]
\label{rem:discrete_theories}
Outside Hardy's framework, the discrete examples we know follow the same
pattern. The stabilizer subtheory has entangled pure states, a state space that
is not a simplex, and even finite basis ambiguity---$\id/4$ is an equal mixture
both of the computational states and of the Bell states---yet it is efficiently
simulable \cite{gottesman1998heisenberg}. Its label permutations are exactly the
affine maps, since a Clifford permutation sends each $Z_j$ to a diagonal Pauli
operator $\pm Z^{\vec{s}_j}$, and for these the degeneracy test is efficient.
Because the Clifford group is, up to phases, a maximal finite subgroup of the
unitary group \cite{Nebe2001clifford}, adjoining any further gate generates an
infinite group whose closure is a Lie group of positive dimension: there is no
discrete theory between the stabilizer subtheory and quantum theory. In
maximally nonlocal theories all reversible dynamics are trivial
\cite{Gross2010reversible}, and polygon theories have finite transformation
groups (Sec.~\ref{sec:gpt}). Entanglement, nonlocality, a non-simplex state space
and finite basis ambiguity each occur in some tractable theory; what none of
them has is a continuous reversible group.
\end{remark}

\section{Discussion and conclusions}
\label{sec:discussion}

Hardy \cite{hardy2001quantum} showed that the sole axiomatic difference between
quantum theory and classical probability is that the former allows continuous
reversible transformations between pure states. The emergence of objective
classicality requires the breakdown of precisely this feature. We have argued
that quantum Darwinism is the physical mechanism realizing that breakdown, and
that the breakdown is visible in the computational structure of concordant
dynamics.

The central results are the following.

\emph{(i)} The obstruction to simulating concordant circuits is the absence of a
uniquely identifiable local product basis, and this absence is theory-level
rather than algorithm-level. A single elementary construction,
Proposition~\ref{prop:entangled_mixture}, exhibits the underlying freedom and
separates it into two regimes: a concordant state can be built out of entangled
pure states, which is why branch tracking fails, and under degeneracy that
alternative preparation becomes orthonormal, giving a rival eigenbasis, which is
$\beta > 0$---which is what defeats basis tracking. Both known
routes---Eastin's state-based algorithm and the gate-based Local Basis
Finder, founder on the same freedom, at different levels of that hierarchy
(Sec.~\ref{sec:common}, Corollary~\ref{cor:two_failures}).

\emph{(ii)} The missing ingredient can be named and located. Basis objectivity
(Definition~\ref{def:basis_objectivity}) decomposes into record classes supplied
by an SBS environment, which is a property of the state, and an interaction map
supplied by $H_{SE}$ being a property of the dynamics. Neither alone suffices
(Fig.~\ref{fig:twoaxes}). In the oracle model of Sec.~\ref{sec:oraclemodel},
supplying both replaces the NP-hard circuit-history test by
$O(d^{2k}\poly(D_F))$ pairwise comparisons per gate, independent of register
size and circuit depth (Proposition~\ref{prop:darwinres}).

\emph{(iii)} The freedom can be graded. The measure $\beta \in [0,1]$ is the
normalized total multiplicity of the degenerate sectors on a gate support, i.e., the
dimension of the manifold of concordance-preserving gate factors the state
admits. The sectors are the operator-valued eigenvalue coincidences identified
in Refs.~\cite{eastin2010simulating,Cable_2015}; what is new is that they are
assigned a graded measure. The threshold is sharp: an admissible gate defeating
the LBF exists if and only if $\beta > 0$
(Theorems~\ref{thm:beta_zero} and \ref{thm:beta_pos},
Corollary~\ref{cor:threshold}), and both directions are constructive. The
measure is invariant and faithful on the basis-structure-preserving operations
(Theorem~\ref{thm:beta_invariance}) and is changed only by non-unitary
operations such as measurement and mixing.

\emph{(iv)} Weakening the continuity requirement of Hardy's fifth axiom to bare
reversibility removes both regimes at once: the state space becomes a simplex,
no state can be built from pure states carrying correlations it does not itself
carry, and $\beta$ vanishes identically
(Theorem~\ref{thm:hardy_collapse}). The converse holds as well, and
discrete-group GPTs carry $\beta_{\mathcal{G}} = 0$ identically
(Sec.~\ref{sec:gpt}). That a separable state can be a mixture of entangled
states is, in this precise sense, the one feature of quantum theory the
obstruction rests on, and $\beta$ is a computational witness of it.

\paragraph{The status of the claim.}
We follow Eastin's own criterion: if computations lacking a property can be
simulated efficiently, that property is a \emph{necessary} resource for
advantage. Theorem~\ref{thm:beta_zero} establishes exactly this for $\beta$, and
we claim no more. It should be emphasized what this does and does not say about
discord. The necessity of discord is established only for concordant
computations of conventional form built from gates of support at most two; for
larger support the question is open, and it is precisely in that regime that the
present measure operates. Should concordant computation with larger gates prove
not to be efficiently simulable, the consequence is that quantum advantage
exists in the total absence of entanglement and discord, and discord is then not
the resource. The purpose of $\beta$ is to name and grade what would be doing
the work instead.

The pure-state literature is the appropriate comparison: entanglement is
necessary for exponential speedup, and the Gottesman--Knill theorem shows it is
not sufficient. There is a further asymmetry worth stating. For pure
states a graded sufficient-side statement exists: Vidal's bound guarantees
efficiency whenever the entanglement is logarithmic---and no comparable
statement is available here. Establishing one, and thereby determining whether
concordant computation can exceed probabilistic reversible classical computation
at all, remains the central question in this area.

\paragraph{Open problems.}
Our work leads to several new directions and open problems. The first is whether basis ambiguity obstructs every efficient
classical strategy or only the basis-tracking ones;
Corollary~\ref{cor:threshold} establishes the existence of a gate defeating one
particular algorithm, and whether $\beta > 0$ implies hardness in any stronger
sense is Conjecture~\ref{conj:main}(ii). The second is the scaling of simulation
cost with $\beta$. As shown in Sec.~\ref{sec:nonmono}, the failure frequency of
the LBF is not monotone in $\beta$: it first rises and then falls as the sectors
grow. Since this is a property of one incomplete algorithm under one
distribution over gates, we draw no conclusion from it, and the correct analogue
of Vidal's $\poly(n)\cdot\poly(\chi)$ bound remains to be found. The third is
whether a genuine monotone exists for a class of operations including the
non-unitary ones, or gates whose support exceeds $b$
(Remark~\ref{rem:non_invariance}). The fourth is whether basis ambiguity can be
put to cryptographic use, which is left open. The fifth
concerns whether the oracle of Sec.~\ref{sec:oraclemodel} can be weakened---in
particular, how much of the interaction map is genuinely needed, and whether a
partially specified $H_{SE}$ still suffices.

Finally, the reliance on a preferred and hidden basis suggests a parallel with
contextuality as a computational resource \cite{Howard2014-tb}. By hiding the
preferred basis one prevents a noncontextual classical simulation from tracking
the state evolution, in the same way that magic states prevent efficient
simulation by stabilizer methods. Whether basis ambiguity implies contextuality,
or instead admits an independent characterization, is a natural direction for
future work.

The broader message is that objectivity is not merely an empirical observation
about the classical world. It is embedded in the computational structure of
quantum theory, and it helps determine which quantum dynamics can and cannot be
efficiently described classically.

\begin{acknowledgments}
 We acknowledge funding by ANRF file number ANRF/ARGM/2025/002679/TS. We acknowledge the use of AI-based language tools, which assisted in improving the clarity and quality of the manuscript’s writing and editing. All scientific content and interpretations remain the authors’ own.
\end{acknowledgments}

\bibliographystyle{apsrev4-2}
\bibliography{bibli2}

\appendix

\section{The Local Basis Finder algorithm}
\label{app:lbf}

For completeness we summarize the algorithm of Ref.~\cite{Cable_2015}. Let
$b \subseteq \{1,\ldots,n\}$ be the set of subsystems on which a gate $G_t$ acts
nontrivially, and let $X_k$ denote computational basis projectors of arbitrary
rank acting on these subsystems. Assume the local basis $U_{t-1}$ at step $t-1$
is known, initially the identity. Define
\begin{equation}
\tilde{X}_k = G_t U_{t-1} X_k U_{t-1}^\dagger G_t^\dagger .
\end{equation}

\emph{Step 1 (commutator condition).} For each $\tilde{X}_k$ and each subsystem
$j \in b$, find all rank-one projectors $\Pi^{(j)}$ satisfying
Eq.~\eqref{eq:commutator}.

\emph{Step 2 ($X_k$-unique local basis).} Let $LB_k^{(j)}$ be the collection of
local rank-one projectors on subsystem $j$ common to all solutions of the
commutator condition for the given $\tilde{X}_k$.

\emph{Step 3 (select the informative set).} Define the subset
$\chi \subseteq \{X_k\}$ of projectors for which a unique local basis solution
exists, i.e.\ for which $LB_k$ is nonempty and consists of a single orthonormal
basis of rank-one projectors.

\emph{Step 4 (compatible local basis).} If such a set exists, define the new
local basis by $U_t = \bigcap_{X_k \in \chi} LB_k$.

\emph{Step 5 (failure condition).} If the intersection in Step 4 is empty, or
contains incompatible local projectors---non-commuting sets---the LBF fails.

When the LBF succeeds it yields a product of local rank-one projectors defining
the local basis at step $t$ in which the state remains concordant. When it fails,
no such basis is determined by the gate, because of the degeneracy-induced
ambiguity analyzed in Sec.~\ref{sec:bam}.

It is often sufficient to apply the algorithm to the rank-one computational
basis projectors $X_k = \proj{\vec{x}_k}$ alone. If a local rank-one projector
satisfies Eq.~\eqref{eq:commutator} for all rank-one $X_k$, it also commutes
with any linear combination of these. Rank-one projectors therefore provide a
sufficient criterion for success when they identify a unique local basis;
higher-rank projectors are required only when ambiguity persists.

\section{An explicit two-qubit gate defeating the LBF}
\label{app:lbf_failure}

We give a concrete two-qubit gate of the form $G = VPB$ for which the LBF fails,
even when applied to computational basis projectors of all ranks. This is the
instance referred to as Case~2 in Sec.~\ref{sec:worked_examples}, where
$\beta = 1/6$.

\paragraph{Construction.}
Take $U_{t-1} = \id_4$ and
\begin{equation}
B = \begin{pmatrix}
1 & 0 & 0 & 0 \\
0 & 1 & 0 & 0 \\
0 & 0 & \cos\theta & -\sin\theta \\
0 & 0 & \sin\theta & \cos\theta
\end{pmatrix},
\qquad \theta = \frac{\pi}{4},
\end{equation}
acting nontrivially only on $\mathrm{span}\{\ket{10},\ket{11}\}$, which is the
degenerate sector of $\rho = \mathrm{diag}(0.4,0.3,0.15,0.15)$. Let $P$ be the
permutation exchanging $\ket{10} \leftrightarrow \ket{11}$, and let
$V = R_z(\pi/4) \otimes H$, with $R_z(\phi) = \mathrm{diag}(e^{-i\phi/2},
e^{i\phi/2})$ and $H$ the Hadamard. The total gate is $G = VPB$. Since $B$
commutes with $\rho$, the gate is concordance-preserving.

\paragraph{Procedure and result.}
Compute $\tilde{X}_k = G X_k G^\dagger$ for all $15$ computational basis
projectors of ranks $1$ through $4$, and for each extract the reduced state on
the second qubit, $\rho_B^{(k)}$, together with its eigenbasis. The rank-one
projectors already suffice:
\begin{align}
X_1 = \proj{00} &\;\Rightarrow\; \rho_B^{(1)} = \proj{+}, \nonumber\\
X_2 = \proj{01} &\;\Rightarrow\; \rho_B^{(2)} = \proj{-}, \nonumber\\
X_3 = \proj{10} &\;\Rightarrow\; \rho_B^{(3)} = \proj{0}, \nonumber\\
X_4 = \proj{11} &\;\Rightarrow\; \rho_B^{(4)} = \proj{1}.
\end{align}
The first two are diagonal in the $X$ basis and the last two in the $Z$ basis.
There is no single orthonormal basis on the second qubit in which all
$\rho_B^{(k)}$ are diagonal, so the retained constraints are incompatible and
the LBF halts.

\paragraph{Role of $B$.}
If the block-diagonal factor is removed, $B = \id$, all reduced states become
diagonal in a common basis---the $X$ basis in this example---and the LBF
succeeds. The failure is therefore attributable entirely to the block-diagonal
freedom, which by Proposition~\ref{prop:concordance_preserving} is admissible
precisely because the state is degenerate on the support. This is the mechanism
that Theorem~\ref{thm:beta_pos} generalizes to arbitrary supports and arbitrary
sector structures.

\section{Sampling procedure for Sec.~\ref{sec:nonmono}}
\label{app:sampling}

The failure frequencies reported in Table~\ref{tab:nonmono} are obtained as
follows. Fix a diagonal two-qubit state $\rho$ and compute its sectors
$\{S_i\}$ from Eq.~\eqref{eq:sectors}. An admissible gate is then drawn as
$G = U P B$, where

\begin{enumerate}
\item $U = V_1 \otimes V_2$ with $V_1, V_2$ independently Haar-random in
$U(2)$;
\item $P$ is drawn uniformly from the $4! = 24$ permutations of the product
basis labels;
\item $B = \bigoplus_i B_i$ with $B_i$ Haar-random in $U(m_i)$ on the sector
$S_i$, so that $[B,\rho] = 0$ by construction and $G$ is
concordance-preserving by
Proposition~\ref{prop:concordance_preserving}.
\end{enumerate}

For each sampled $G$ the LBF of Appendix~\ref{app:lbf} is run on all $15$
computational basis projectors with $U_{t-1} = \id$. A projector is discarded
when $\tilde{X}_k$ is entangled, i.e.\ when no product of local rank-one
projectors commutes with it; it is retained otherwise, and the demanded
single-subsystem bases are recorded. The gate is counted as \emph{defeating} the
LBF when two retained constraints on the same subsystem demand non-commuting
rank-one projectors. Entanglement of $\tilde{X}_k$ and commutation are tested
numerically to fixed tolerance.

Two points about the procedure bear on the interpretation of
Table~\ref{tab:nonmono}. First, the distribution over gates is a choice, not a
canonical object: the Haar measure on each sector is the natural uniform choice
but has no operational privilege, and the structured gates of
Theorem~\ref{thm:beta_pos} form a measure-zero family under it. Second, the
discarding step in the LBF is what produces the non-monotonicity: on sectors of
size three or more a Haar-random $B_i$ almost surely entangles rank-one product
projectors, so almost every constraint is discarded and no conflict can arise.
The table therefore measures the sensitivity of a particular algorithm to a
particular measure, and nothing about the state beyond what
Corollary~\ref{cor:threshold} already establishes.

\end{document}